\documentclass[10pt,reqno]{amsart}

\RequirePackage[OT1]{fontenc}
\usepackage{graphicx,amsmath,amssymb,epsfig}
\long\def\comment#1{}
\RequirePackage[OT1]{fontenc}
\RequirePackage{amsthm,amsmath}
\RequirePackage[colorlinks,citecolor=blue,urlcolor=blue]{hyperref}
\RequirePackage{hypernat}
\usepackage{graphicx,amsmath,amssymb,epsfig}

\newtheorem{theorem}{Theorem}

\newtheorem{corollary}{Corollary}
\newtheorem{lemma}{Lemma}

\theoremstyle{definition}

\def\LASSO{\mbox{square-root lasso}}
\def\be{\begin{eqnarray}}
\def\ee{\end{eqnarray}}

\def\ba{\begin{array}}
\def\ea{\end{array}}
\def\bs{\begin{align}\begin{split}\nonumber}
\def\bsnumber{\begin{align}\begin{split}}
\def\es{\end{split}\end{align}}

\def\({\left(}
\def\){\right)}
\def\[{\left[}
\def\]{\right]}
\def\hat{\widehat}
\def\leq{\leqslant}
\def\geq{\geqslant}

\def\cc{\bar{c}}

\def\Pr{{\rm pr}}

\def\En{E_n}
\def\Ep{E}%{\mathbf{E}}

\def\RR{{\rm I\kern-0.18em R}}

\def\supp{{\rm supp}}

\def\bsqrt#1{\{ #1 \}^{1/2}}
\def\bbsqrt#1{\left\{ #1 \right\}^{1/2}}

\begin{document}

%% Here are the title, author names and addresses
\title[Square-root lasso]{Square-root lasso: pivotal recovery of sparse signals via conic programming}

\author{Alexandre Belloni, Victor Chernozhukov and Lie Wang}

\maketitle

\begin{abstract}
We propose a pivotal method for estimating high-dimensional sparse
linear regression models, where the overall number of regressors
$p$ is large, possibly much larger than $n$, but only $s$
regressors are significant. The method is a modification of the lasso,
called the square-root lasso.  The method is pivotal in that it neither relies on the
knowledge of the standard deviation $\sigma$ or nor does it need to pre-estimate $\sigma$.
Moreover, the method does not rely on normality or sub-Gaussianity of noise.  It achieves  near-oracle
performance, attaining the convergence rate $\sigma
\{(s/n)\log p\}^{1/2}$ in the prediction norm, and thus matching  the performance of the
lasso with known $\sigma$.  These performance results are valid for both Gaussian and non-Gaussian errors, under
some mild moment restrictions. We formulate the square-root lasso as a solution to a
convex conic programming problem, which allows us to implement the estimator using
efficient algorithmic methods, such as
interior-point and first-order methods.
\end{abstract}

%\begin{keywords}
%conic programming; high-dimensional sparse model; moderate deviation theory.
%\end{keywords}

\section{Introduction}

We consider the linear regression model for outcome $y_i$ given fixed $p$-dimensional regressors $x_i$:
%\begin{eqnarray*} & & y_i = x_i'\beta_0 + \sigma \epsilon_i, \ \ i=1,...,n,\end{eqnarray*}
\begin{equation}\label{Def:Reg}
y_i = x_i'\beta_0 + \sigma \epsilon_i \ \ (i=1,...,n)
\end{equation}
with independent and identically distributed noise $\epsilon_i$ $(i =1,...,n)$ having law $F_0$ such that
\begin{equation}\label{Def:Error}
\Ep_{F_0} (\epsilon_i) = 0 \ , \ \Ep_{F_0}(\epsilon_i^2) =1.
\end{equation}
 The vector $\beta_0 \in \Bbb{R}^p$ is the unknown true parameter value, and $\sigma>0$ is the unknown standard deviation. The regressors $x_i$ are $p$-dimensional,
$
x_i = (x_{ij}, j=1,...,p)',
$
where the dimension $p$ is possibly much larger than the sample
size $n$. Accordingly, the true parameter value $\beta_0$ lies in
a very high-dimensional space $\Bbb{R}^p$.  However, the key
assumption that makes the estimation possible is the sparsity of $\beta_0$:
\begin{equation}\label{Def:T}
T = \supp(\beta_0) \text{ has } s<n \text{ elements}.
\end{equation}
The identity $T$ of the significant regressors is unknown.
Throughout, without loss of generality, we normalize  \begin{equation}\label{Def:Normalize}
 \frac{1}{n} \sum_{i=1}^n
x_{ij}^2 = 1 \ \ \ (j=1,\ldots,p).\end{equation} In making asymptotic
statements below we allow for $s\to \infty$ and  $p \to \infty$ as
$n \to \infty$.

The ordinary least squares estimator is not consistent
for estimating $\beta_0$ in the setting with $p > n$. The lasso estimator \cite{T1996} can restore consistency under mild conditions by penalizing through the sum of absolute parameter
values:\begin{equation}\label{Def:LASSOmain} \bar \beta \in \arg
\min_{\beta \in \Bbb{R}^p} \widehat Q (\beta) +  \frac{\lambda}{n}
 \| \beta \|_{1},
\end{equation}
where $ \widehat Q (\beta) = n^{-1} \sum_{i=1}^n (y_i - x_i'\beta)^2$ and $\|\beta\|_{1} = \sum_{j=1}^p | \beta_j|$.  The lasso estimator is computationally attractive because it minimizes a structured convex function.
 Moreover, when errors are normal, $F_0= N(0,1)$, and suitable design conditions hold, if one uses the penalty level
\begin{equation}\label{Def:LASSOLambda}
\lambda = \sigma  c  2 n^{1/2}  \Phi^{-1}(1-\alpha/2p)
\end{equation}
for some constant $c>1$, this estimator achieves the near-oracle performance, namely
 \begin{equation}\label{bound one}
\| \bar \beta - \beta_0 \|_2 \lesssim \sigma \left\{ s \log (2p/\alpha)/n\right\}^{1/2},
 \end{equation}
with probability at least $1-\alpha$.  Remarkably, in  (\ref{bound one}) the overall
number of regressors $p$  shows up
only through a logarithmic factor, so that if $p$ is polynomial in
$n$, the oracle rate is achieved up to a factor of $\log^{1/2} n$.
Recall that the oracle knows the identity $T$ of significant
regressors, and so it can achieve the rate $\sigma (s/n)^{1/2}$.
Result (\ref{bound one}) was demonstrated by \cite{BickelRitovTsybakov2009},
and closely related results were given in \cite{MY2007}, and
\cite{ZhangHuang2006}.
\cite{CandesTao2007}, \cite{vdGeer}, \cite{Koltchinskii2009}, \cite{BuneaTsybakovWegkamp2007}, \cite{ZhaoYu2006},
\cite{HHS2008}, \cite{Wainright2006}, and \cite{Zhang2009} contain other fundamental results obtained for
related problems; see \cite{BickelRitovTsybakov2009} for
further references.

Despite these attractive features, the lasso construction
(\ref{Def:LASSOmain}) -- (\ref{Def:LASSOLambda})  relies on
knowing the standard deviation $\sigma$ of the noise. Estimation
of $\sigma$ is non-trivial when $p$ is large, particularly when $p \gg n$, and
remains an outstanding practical and theoretical problem. The
estimator we propose in this paper, the square-root lasso, eliminates the need to
know or to pre-estimate $\sigma$. In addition, by  using moderate
deviation theory, we can dispense with the normality assumption $F_0 = \Phi$ under certain conditions.

The $\LASSO$ estimator of $\beta_0$ is defined as the solution to the optimization problem
\begin{equation}\label{Def:LASSOmod}
\hat \beta \in  \arg\min_{\beta \in \RR^p} \{\hat Q(\beta)\}^{1/2} +
\frac{\lambda}{n}  \|\beta\|_1,\end{equation}
with the penalty level
\begin{equation}\label{our penalty}
\lambda = c  n^{1/2}  \Phi^{-1}(1-\alpha/2p),
 \end{equation}
for some constant $c>1$. The
penalty level in (\ref{our penalty}) is independent of $\sigma$, in contrast to (\ref{Def:LASSOLambda}), and hence is pivotal with respect
to this parameter.  Furthermore, under reasonable conditions, the proposed penalty level (\ref{our penalty}) will also be valid  asymptotically without imposing normality $F_0 = \Phi$, by virtue of  moderate deviation theory.
%(this side result is new and is of independent interest for other
%problems with $\ell_1$-penalization).

We will show that the $\LASSO$ estimator achieves the near-oracle rates of
convergence under suitable design conditions and suitable
conditions on $F_0$ that extend significantly beyond normality:
 \begin{equation}
\| \widehat \beta - \beta \|_2 \lesssim \sigma \left\{s \log (2p/\alpha)/n\right\}^{1/2} ,
 \end{equation}
with probability approaching $1-\alpha$. Thus, this estimator matches the near-oracle
performance of lasso, even though the noise level $\sigma$ is unknown.  This is the main result
of this paper.  It is important to emphasize here that this result is not a direct consequence
of the analogous result for the lasso. Indeed, for a given value of the penalty level, the statistical structure
of the $\LASSO$   is different from that of the lasso, and so our proofs are also different.

Importantly, despite taking the square-root of the least squares criterion function, the problem (\ref{Def:LASSOmod}) retains global convexity, making
the estimator computationally attractive. The second main result of this paper
is to formulate the $\LASSO$ as a solution to a conic programming  problem.
Conic programming can be seen as linear programming  with conic constraints,
so it generalizes canonical linear programming with non-negative orthant constraints, and inherits a rich set of theoretical properties and algorithmic methods from linear programming. In our case, the constraints
take the form of a second-order cone, leading to a particular,  highly tractable, form of conic programming. In turn, this allows us to implement the estimator using
efficient algorithmic methods, such as
interior-point methods, which provide polynomial-time
bounds on computational time \cite{NesterovNemirovski,RenegarBook}, and modern first-order methods
\cite{Nesterov2005,Nesterov2007,LanLuMonteiro,BeckTeboulle2009}.

 In what follows, all true parameter values, such as $\beta_0$, $\sigma$, $F_0$, are
implicitly indexed by the sample size $n$, but we omit the index in our notation whenever
this does not cause confusion. The regressors $x_i$ $(i =1,\ldots,n)$
are taken to be fixed throughout. This includes random design as a
special case, where we condition on the realized values of the
regressors. In making asymptotic statements, we assume that $n \to
\infty$ and $p=p_n  \to \infty$, and we also allow for $s=s_n \to
\infty$. The notation $o(\cdot)$ is defined with respect to $n\to \infty$. We use the notation $(a)_+ = \max(a,0)$, $a \vee b =
\max(a, b)$ and $a \wedge b = \min(a,b)$. The
$\ell_2$-norm is denoted by $\|\|_2$, and $\ell_\infty$ norm by $\|\|_{\infty}$.
Given a vector $\delta \in \RR^p$ and a set of indices $T
\subset \{1,\ldots,p\}$, we denote by $\delta_T$ the vector in
which $\delta_{Tj} = \delta_j$ if $j\in T$, $\delta_{Tj}=0$ if
$j\notin T$. We also use $\En(f)= \En\{f(z)\} = \sum_{i=1}^n
f(z_i)/n$. We use $a \lesssim b$ to denote $a \leqslant c b$ for
some constant $c>0$ that does not depend on $n$.

\section{The choice of penalty level}

\subsection{The general principle and heuristics}

The key quantity determining the choice of the penalty level for $\LASSO$ is the score, the gradient of $\hat Q^{1/2}$
evaluated at the true parameter value $\beta = \beta_0$:
$$ \tilde S= \nabla \widehat Q^{1/2}(\beta_0) =  \frac{ \nabla{\widehat Q(\beta_0)}}{ 2 \bsqrt{\widehat Q(\beta_0)}}=\frac{\En(x \sigma \epsilon)}{\bsqrt{\En(\sigma^2 \epsilon^2)}}=  \frac{\En(x \epsilon)}{\bsqrt{\En(\epsilon^2)}}. $$
The score $\tilde S$ does not depend on the unknown standard deviation $\sigma$ or the unknown true parameter value $\beta_0$,
and therefore is pivotal with respect to $(\beta_0,\sigma)$.  Under the classical normality assumption, namely
$F_0 = \Phi$, the score is in fact completely pivotal, conditional on $X$.  This means that in principle we know the distribution of $\tilde S$ in this case, or at least we can compute it by simulation.

The score $\tilde S$ summarizes the estimation noise in our problem, and we may set the penalty level  $\lambda/n$
to overcome it. For reasons of efficiency,
we set $\lambda/n$ at the smallest level that dominates the estimation noise, namely we choose the smallest $\lambda$ such that
\begin{equation}\label{BRT principle}
\lambda \geq c\Lambda,  \ \ \ \Lambda = n\| \tilde
S\|_{\infty}, \ \ \
\end{equation}
with a high probability, say $1- \alpha$, where $\Lambda$ is the
maximal score scaled by $n$, and $c>1$ is a theoretical constant of
\cite{BickelRitovTsybakov2009} to be stated later. The principle
of setting $\lambda$ to dominate the score of the criterion
function is motivated by \cite{BickelRitovTsybakov2009}'s choice
of penalty level for the lasso. This general principle carries over to other convex problems, including ours, and
that leads to the optimal, near-oracle, performance of other $\ell_1$-penalized estimators.

 In the case of the $\LASSO$ the maximal score is pivotal, so the
penalty level in (\ref{BRT principle}) must also be pivotal.
We used the square-root transformation %$\bsqrt{\hat Q}$
in the $\LASSO$ formulation (\ref{Def:LASSOmod}) precisely to
guarantee this pivotality.  In contrast, for lasso,
%which uses $\hat Q$,
the score $S = \nabla \hat Q(\beta_0) =
2\sigma  \En (x\epsilon) $ is obviously non-pivotal, since it
depends on $\sigma$. Thus, the penalty level for lasso
must be non-pivotal.  These theoretical differences translate
into  obvious practical differences.  In the lasso,
we need to guess conservative upper bounds $\bar \sigma$ on $\sigma$, or we need to
 use preliminary estimation of $\sigma$ using a pilot lasso, which uses
 a conservative upper bound $\bar \sigma$ on $\sigma$.  In the $\LASSO$, none of these is needed.   Finally, the use of pivotality principle for
constructing the penalty level is also fruitful in other problems with pivotal scores, for example, median regression \cite{BC-SparseQR}.

The rule (\ref{BRT principle}) is not practical, since we do
not observe $\Lambda$ directly. However, we can proceed as follows:

1.   When we know the distribution of errors exactly,
e.g., $F_0 = \Phi$, we propose to set $\lambda$ as $c$ times
the $(1-\alpha)$ quantile of $\Lambda$ given $X$.
This choice of the penalty level precisely
implements (\ref{BRT principle}), and is easy to compute by
simulation.

2.  When we do not know $F_0$ exactly, but instead
know that $F_0$ is an element of some family $\mathcal{F}$, we can
rely on either finite-sample or asymptotic upper bounds on
quantiles of $\Lambda$ given $X$. For example, as mentioned
in the introduction, under some mild conditions on $\mathcal{F}$,  $\lambda = c  n^{1/2}
\Phi^{-1}(1-\alpha/2p)$ is a valid asymptotic choice.

What follows below elaborates these approaches. Before describing the details, it is useful to mention some heuristics for the principle (\ref{BRT principle}). These arise from considering the simplest case, where none of the regressors are significant, so that $\beta_0 =0$.
We want our estimator to perform at a near-oracle level in all cases, including this one. Here the oracle estimator is $\tilde \beta= \beta_0 =0$.  We also want $\widehat \beta = \beta_0 = 0$ in this case, at least with a high probability $1-\alpha$.  From the subgradient optimality conditions of (\ref{Def:LASSOmod}), in order for this to be true we must have
$
- \tilde S_j + \lambda/n \geq 0 \text{ and } \tilde S_j + \lambda/n \geq 0$ $(j=1,\ldots,p)$.
We can only guarantee this by setting the penalty level $\lambda/n$ such that
$
\lambda \geq n \max_{1 \leq j\leq p} |\tilde S_j| = n \|\tilde S\|_{\infty}$ with probability at least $1-\alpha$.
This is precisely the rule (\ref{BRT principle}), and,  as it turns out, this delivers near-oracle performance more generally, when $\beta_0 \neq 0$.

\subsection{The formal choice of penalty level and its properties} In order to describe our choice of $\lambda$ formally, define
 for $0 < \alpha < 1$
\begin{eqnarray}
&&\Lambda_F(1-\alpha\mid X) =  (1-\alpha) \text{--quantile of } \Lambda_F \mid X,  \label{LambdaDefF} \\
&& \Lambda(1-\alpha) =  n^{1/2} \Phi^{-1}(1-\alpha/2p) \leq \bsqrt{2 n \log(2p/\alpha)},\label{LambdaDefAsmp}
\end{eqnarray}
where $\Lambda_F = n\|\En(x\xi)\|_\infty/\{\En(\xi^2)\}^{1/2}$, with independent and identically distributed $\xi_i$ $(i=1,\ldots,n)$ having law $F$. We can compute (\ref{LambdaDefF}) by simulation.

In the normal case, $F_0 = \Phi$,  $\lambda$ can be either of
%\begin{eqnarray}\label{lambda: normal}
%& \lambda = c \Lambda_\Phi(1-\alpha\mid X), \\
% \label{lambda: normal asy}
%& \lambda = c \Lambda(1-\alpha) = c   n^{1/2} \Phi^{-1}(1-\alpha/2p),
%\end{eqnarray}
\begin{equation}\label{normal}
\begin{array}{cl}
& \lambda = c \Lambda_\Phi(1-\alpha\mid X), \\
& \lambda = c \Lambda(1-\alpha) = c   n^{1/2} \Phi^{-1}(1-\alpha/2p),
\end{array}
\end{equation}
 which we call here the exact and asymptotic options, respectively. The parameter $1-\alpha$ is a confidence level which guarantees near-oracle performance with probability $1-\alpha$; we recommend
$1-\alpha = 0.95$.
The constant  $c>1$ is a theoretical constant of \cite{BickelRitovTsybakov2009}, which is needed to guarantee
a regularization event introduced in the next section; we recommend $c=1.1$. The options in (\ref{normal}) are valid either in finite or large samples under the conditions stated below. They are also supported by the finite-sample experiments reported in Section \ref{Sec:Empirical}. We recommend using the exact option over the asymptotic option, because by construction the former is better tailored to the given sample size $n$ and design matrix $X$. Nonetheless, the asymptotic option is easier to compute. Our theoretical results in section 3 show that the options in (\ref{normal}) lead to near-oracle rates of convergence.

For the asymptotic results, we shall impose the following condition: \\

\textsc{Condition G.} \textit{We have that $\log^2(p/\alpha)\log(1/\alpha) = o(n)$ and $p/\alpha \to \infty$ as $n\to\infty$.} \\

The following lemma shows that the exact and asymptotic options in (\ref{normal}) implement the regularization event $\lambda \geq c \Lambda$
in the Gaussian case with the exact or asymptotic probability $1-\alpha$ respectively.  The lemma also bounds the magnitude of the penalty level
for the exact option, which will be useful for stating bounds on the estimation error.  We assume throughout the paper that $0< \alpha <1$ is bounded
away from 1, but we allow $\alpha$ to approach $0$ as $n$ grows.

\begin{lemma} \label{Lemma: normal} Suppose that $F_0=\Phi$. (i) The exact option in (\ref{normal}) implements $\lambda \geq c \Lambda$ with probability at least $1-\alpha$. (ii) Assume that $p/\alpha>8$. For any $1<\ell<\bsqrt{n/\log(1/\alpha)}$, the asymptotic option in (\ref{normal}) implements $\lambda
\geq c \Lambda$ with probability at least
$$1-\alpha  \tau, \ \  \ \tau = \left\{1+\frac{1}{\log(p/\alpha)}\right\}\frac{\exp[2\log(2p/\alpha)\ell\bsqrt{\log(1/\alpha)/n}]}{1-\ell\bsqrt{\log(1/\alpha)/n}}-\alpha^{\ell^2/4-1},$$
where, under Condition G, we have $\tau = 1 + o(1)$ by setting $\ell\to\infty$ such that  $\ell = o[n^{1/2}/\{\log(p/\alpha)\log^{1/2}(1/\alpha)\}]$ as $n \to \infty$. (iii) Assume that $p/\alpha>8$ and $n>4\log(2/\alpha)$. Then  $$
\Lambda_\Phi(1-\alpha\mid X) \leq \nu \Lambda(1-\alpha) \leq \nu\bsqrt{2 n
\log(2p/\alpha)}, \  \ \nu= \frac{\bsqrt{1+2/\log(2p/\alpha)}}{1-2\bsqrt{\log(2/\alpha)/n}},$$
where under Condition G, $\nu=1+o(1)$ as $n \to \infty$.
\end{lemma}

In the non-normal case, $\lambda$ can be any of
%\begin{eqnarray}
%& \lambda = c \Lambda_{F}(1-\alpha\mid X), \label{lambda: nonnormal exact}\\
%& \lambda = c  \max_{F \in \mathcal{F}}\Lambda_F(1-\alpha\mid X), \label{lambda: nonnormal semi exact}\\
%&
%\lambda = c \Lambda(1-\alpha) = c   n^{1/2} \Phi^{-1}(1-\alpha/2p) \label{lambda: nonnormal asy},
%\end{eqnarray}
\begin{equation}\label{non-normal}
\begin{array}{cl}
& \lambda = c \Lambda_{F}(1-\alpha\mid X),\\
& \lambda = c  \max_{F \in \mathcal{F}}\Lambda_F(1-\alpha\mid X),\\
&
\lambda = c \Lambda(1-\alpha) = c   n^{1/2} \Phi^{-1}(1-\alpha/2p),
\end{array}
\end{equation}
which we call the exact, semi-exact, and asymptotic options, respectively. We set the confidence level $1-\alpha$ and the constant $c>1$  as
before.  The exact option is applicable when $F_0=F$, as for example in the previous normal case. The semi-exact
option is applicable when $F_0$ is a member of some
family $\mathcal{F}$,  or whenever the family $\mathcal{F}$ gives
a more conservative penalty level. We also assume that
$\mathcal{F}$ in (\ref{non-normal}) is either finite or, more
generally, that the maximum in (\ref{non-normal}) is well defined.
For example, in applications, where the regression errors
$\epsilon_i$ are thought of having  a potentially wide range of
tail behavior, it is useful to set $\mathcal{F} = \{t(4), t(8),
t(\infty)\}$ where $t(k)$ denotes the Student distribution with $k$ degrees of freedom. As stated previously, we can compute the quantiles
$\Lambda_F(1-\alpha\mid X)$ by  simulation. Therefore, we can
implement the exact option easily, and if $\mathcal{F}$ is not too
large, we can  also implement the semi-exact option  easily.
Finally, the asymptotic option is applicable when $F_0$ and design
$X$ satisfy Condition M and has the advantage of being trivial to compute.

%We recommend using either the exact or semi-exact options, when applicable, since they  are better tailored to the given design $X$ and sample size $n$.    However, the asymptotic option also has attractive features: It is  trivial to  compute, and it often provides a very similar, albeit slightly more conservative, penalty level, as confirmed by our numerical experiments.

%\includegraphics[width=0.65\textwidth]{SQRT_LASSO_SIM_LAMBDAnew.eps}
%\captionsetup{width=14cm}SQRTLASSOSIMLAMBDAnew.eps

%\begin{figure}[!h]
%\centering
%\includegraphics[width=0.65\textwidth]{SQRTLASSOSIMLAMBDAnew.eps}%{SQRT_LASSO_SIM_LAMBDAnew.eps}
%\color{white}{\figurebox{-0.5pc}{35pc}{}}
%\color{black}{}
%\caption{The asymptotic bound  $\Lambda(0.95)$ and realized values of $\Lambda(0.95|X,F)$ sorted in increasing order,
% shown for 100 realizations of $X = [x_1,...,x_n]'$, and for three cases of  error distribution, $F=t(4), t(8)$, and $t(\infty)$.
%  For each realization, we have generated $X=[x_1,...,x_n]'$ as independent and identically distributed draws of $x \sim N(0,\Sigma)$ with %$\Sigma_{jk} = (1/2)^{|j-k|}$.}\label{Fig:Penalty}
%\end{figure}

 For the asymptotic results in the non-normal case, we impose the following moment conditions. \\

\textsc{Condition M.} \textit{There exist a finite constant $q> 2$ such that the law $F_0$ is an element of the family $\mathcal{F}$
 such that $ \sup_{n\geq 1}\sup_{F \in \mathcal{F}}  \Ep_F (|\epsilon|^{q}) < \infty;$ the design $X$ obeys $ \sup_{n \geq 1, 1 \leq j \leq p} \En  (|x_j|^q ) < \infty.$} \\

We also have to restrict the growth of $p$ relative to $n$, and we also assume that $\alpha$ is either bounded away from zero or approaches zero not too rapidly. See also the Supplementary Material for an alternative condition. \\

\textsc{Condition R.} \textit{ As $n\to\infty$,  $p\leq \alpha n^{\eta(q-2)/2}/2$  for some constant $0< \eta < 1$, and $\alpha^{-1} = o[ n^{ \{(q/2-1) \wedge (q/4)\} \vee (q/2-2)}/(\log n)^{q/2}]$,  where $q>2$ is defined in Condition M.} \\

The following lemma shows that the options (\ref{non-normal}) implement the regularization event $\lambda \geq c \Lambda$
in the non-Gaussian case with exact or asymptotic probability $1-\alpha$. In particular, Conditions R and M, through relations (\ref{Eq:PRalpha1}) and (\ref{Eq:PRalpha2}), imply that for any fixed $v>0$, \begin{equation}\label{non-gauss:alpha}\Pr \{ |\En(\epsilon^2)- 1| > v\} = o(\alpha), \ \ \  \ n\to\infty.\end{equation}
 The lemma also bounds the magnitude of the penalty level $\lambda$
for the exact and semi-exact options, which is useful for stating bounds on the estimation error in section 3.

\begin{lemma}\label{Lemma: non-normal} (i) The exact option in (\ref{non-normal}) implements $\lambda \geq c \Lambda$ with probability at least $1-\alpha$, if $F_0 = F$.  (ii) The semi-exact option in (\ref{non-normal}) implements $\lambda \geq c \Lambda$ with probability at least $1-\alpha$, if either $F_0 \in \mathcal{F}$ or $\Lambda_F(1-\alpha\mid X)\geq \Lambda_{F_0}(1-\alpha\mid X)$ for some $F \in \mathcal{F}$.  Suppose further that Conditions M and R hold. Then, (iii) the asymptotic option in (\ref{non-normal}) implements $\lambda \geq c \Lambda$ with probability at least $1-\alpha - o(\alpha)$, and
(iv) the magnitude of the penalty level of the exact and
semi-exact options in (\ref{non-normal}) satisfies the inequality $$ \max_{F \in \mathcal{F}}\Lambda_F(1-\alpha\mid X) \leq \Lambda(1-\alpha)\{1+o(1)\} \leq  \bsqrt{2 n \log(2p/\alpha)}\{1+o(1)\}, n \to \infty.$$
\end{lemma}

Thus all of the asymptotic conclusions reached in Lemma \ref{Lemma: normal} about the penalty
level in the Gaussian case continue to hold in the non-Gaussian
case, albeit under more restricted conditions on the growth of $p$
relative to $n$. The growth condition depends on the number
of bounded moments $q$ of regressors and the error terms: the
higher $q$ is, the more rapidly $p$ can grow with $n$. We emphasize
that Conditions M and R are only one possible set of
sufficient conditions that guarantees the Gaussian-like
conclusions of Lemma \ref{Lemma: non-normal}. We derived them using the moderate
deviation theory of \cite{slastnikov}. For example, in the Supplementary Material, we  provide an alternative condition, based on the use of the self-normalized moderate deviation theory of \cite{JingShaoWang2003}, which results in much weaker growth condition
on $p$ in relation to $n$, but requires much stronger conditions on the moments
of regressors.

\section{Finite-sample and asymptotic bounds on the estimation error}

\subsection{Conditions on the Gram matrix}  We shall state convergence rates for $ \hat \delta = \hat \beta - \beta_0$
in the  Euclidean norm $\|\delta\|_2 = (\delta'\delta)^{1/2}$ and also in the prediction norm $$
\|\delta\|_{2,n} = [ \En\{(x'\delta)^2\} ]^{1/2} = \bsqrt{\delta '\En(xx') \delta }.$$
The latter norm directly depends on the Gram matrix $\En (xx')$.   The choice of penalty level described in Section 2 ensures
 the regularization event $\lambda \geq c \Lambda$, with probability $1-\alpha$ or with probability approaching $1-\alpha$.
This event will in turn imply another regularization event, namely that $\hat \delta$ belongs to the restricted set $\Delta_{\cc}$, where
$$\Delta_{\cc} = \{\delta \in \Bbb{R}^p: \|\delta_{T^c}\|_{1} \leq \cc  \| \delta_{T}\|_{1}, \delta \neq 0\},  \  \ \cc = \frac{c+1}{c-1}.$$

Accordingly, we will state the bounds on  estimation errors $\|\hat \delta\|_{2,n}$  and $\|\hat \delta\|_2$
in terms of the following restricted eigenvalues of the Gram matrix $\En (xx')$:
 \begin{eqnarray}\label{RE}
    \kappa_{\cc}  = \min_{\delta \in \Delta_{\cc}} \frac{s^{1/2}\|\delta\|_{2,n}}{\|\delta_T\|_{1} },  \ \  \ \  \widetilde \kappa_{\cc}  = \min_{ \delta \in \Delta_{\cc}
 } \frac{\|\delta\|_{2,n}}{\|\delta\|_2}.
\end{eqnarray}
These restricted eigenvalues can depend on $n$ and $T$, but we suppress the dependence in our notation.

In making simplified asymptotic statements, such as those appearing in Section 1, we invoke the following condition on the restricted eigenvalues:  \\

\textsc{Condition RE.} \textit{ There exist finite constants $n_0>0$ and $\kappa> 0$, such that the restricted eigenvalues obey $\kappa_{\cc} \geq \kappa \text{ and } \tilde \kappa_{\cc} \geq \kappa$ for all  $n>n_0$.} \\

The restricted eigenvalues (\ref{RE}) are simply variants of the  restricted eigenvalues introduced in \cite{BickelRitovTsybakov2009}.  Even though the minimal eigenvalue of the Gram matrix
$\En(xx')$ is zero whenever $p \geq n$, \cite{BickelRitovTsybakov2009} show that its restricted eigenvalues can be
bounded away from zero, and they and others provide sufficient primitive conditions that cover many fixed and random designs of interest, which allow for reasonably general, though not arbitrary, forms of correlation between regressors. This makes conditions on restricted eigenvalues useful for many applications. Consequently, we take the restricted eigenvalues as primitive quantities and Condition RE as primitive. The restricted eigenvalues are tightly tailored to the $\ell_1$-penalized estimation problem. Indeed, $\kappa_{\cc}$ is the modulus of continuity between the estimation norm and the penalty-related term computed over the restricted set,  containing the deviation of the estimator from the true value; and $\tilde \kappa_{ \cc}$ is the modulus of continuity
between the estimation norm and the Euclidean norm over this set.

It is useful to recall at least one simple sufficient condition
for bounded restricted eigenvalues. If for $m = s\log n$, the
$m$--sparse eigenvalues of the Gram matrix $\En(xx')$ are bounded
away from zero and from above for all $n > n'$, i.e.,
\begin{equation}\label{SE}
0 < k \leq \min_{\|\delta_{T^c}\|_{0} \leq m, \delta \neq 0
 } \frac{ \|\delta\|_{2,n}^2}{\|\delta\|^2_2} \leq  \max_{\|\delta_{T^c}\|_{0} \leq m, \delta \neq 0
 } \frac{ \|\delta\|^2_{2,n}}{\|\delta\|^2_2} \leq k' < \infty,
\end{equation}
for some positive finite constants $k$, $k'$, and $n'$, then Condition RE holds
once $n$ is sufficiently large. In words, (\ref{SE}) only requires
the eigenvalues of certain small $m \times m$ submatrices of the large $p\times p$ Gram matrix to be bounded from above and below.
The sufficiency of (\ref{SE}) for Condition RE follows from \cite{BickelRitovTsybakov2009}, and  many sufficient conditions for (\ref{SE}) are provided by
 \cite{BickelRitovTsybakov2009}, \cite{ZhangHuang2006}, \cite{MY2007}, and \cite{RudelsonVershynin2008}.

%Indeed, under the primitive conditions RSE.1$(m)$ and RSE.2$(m)$, with $m=s\log n$, we can bound $\kappa_{\cc}$ from below
%$$ \kappa_1  \geq \bsqrt{\kappa(m)^2}\( 1 - \cc\bsqrt{s\phi(0)/[ m \kappa(m)^2 ]}\)\geq \kappa(m) \( 1 -  \mmu{m} \cc\bsqrt{s/m}\).$$
%$$\begin{array}{rcl}
% \kappa_{\cc} & \geq & \kappa(m) \( 1 -  \mmu{m} {\cc}\bsqrt{s/m}\)
%  =  \kappa(s\log n) \( 1 -  \mmu{(s\log n)} {\cc}\bsqrt{1/\log n}\) \end{array}$$
%%%%%%%%%%%%%%%
%%%%%%%%%%%%%%%

\subsection{Finite-sample  and asymptotic bounds on estimation error}

%Next we prove the rate of convergence for the square root LASSO.
%The proof is substantively different from the proof for LASSO, since
%the penalty level does not depend on the noise level.

We now present the main result of this paper. Recall that we do not assume that the noise is sub-Gaussian or that $\sigma$ is known.

\begin{theorem}\label{Thm:RateSquareRootLASSONonparametric}
Consider the model described in (\ref{Def:Reg})--(\ref{Def:Normalize}). Let $c>1$, $\cc=(c+1)/(c-1)$, and suppose that $\lambda$ obeys the growth restriction  $\lambda s^{1/2} \leq n\kappa_{\cc} \rho$, for some $\rho<1$.
If $\lambda \geq c \Lambda$, then
$$
\|\hat \beta - \beta_0\|_{2,n}  \leq A_n \sigma \bsqrt{\En(\epsilon^2)} \frac{\lambda s^{1/2}}{n},  \ \ A_n =\frac{2( 1 + 1/c )}{\kappa_{\cc}(1-\rho^2)}.
$$
In particular, if $\lambda \geq c \Lambda$ with probability at least $1-\alpha$,
and $\En (\epsilon^2) \leq  \omega^2$ with probability at least $1- \gamma$, then with probability at least $1-\alpha - \gamma$,
$$
\tilde \kappa_{\cc} \|\hat \beta - \beta_0\|_2 \leq \| \hat \beta - \beta_0 \|_{2,n}  \leq A_n \sigma \omega \frac{\lambda s^{1/2}}{n}.
$$
\end{theorem}

This result provides a finite-sample bound for $\hat \delta$ that is similar to that for the lasso
estimator with known $\sigma$, and this result leads to the same rates of convergence as in the case of lasso.
It is important to note some differences. First,
for a given value of the penalty level $\lambda$, the statistical structure of $\LASSO$   is  different from that of lasso,
and so our proof of Theorem \ref{Thm:RateSquareRootLASSONonparametric} is also different. Second, in the proof
we have to invoke the additional growth restriction,
$\lambda s^{1/2} < n\kappa_{\cc}$, which is not present in the lasso analysis that treats $\sigma$ as known.
We may think of this restriction as the price of not knowing $\sigma$ in our framework.
However, this additional condition is very mild and holds asymptotically under typical conditions if  $ (s/n) \log (p/\alpha) \to 0$,
as the corollaries below indicate, and it is independent of $\sigma$.
In comparison, for the lasso estimator, if we treat $\sigma$ as unknown and attempt to estimate it consistently using a pilot lasso, which uses an upper bound $\bar \sigma \geq \sigma$ instead
of $\sigma$, a similar growth condition $(\bar \sigma/\sigma) (s/n) \log (p/\alpha) \to 0$ would have to be imposed, but this condition depends on $\sigma$ and is more restrictive than our growth condition when $\bar \sigma/\sigma$ is large.

Theorem \ref{Thm:RateSquareRootLASSONonparametric}  implies the following bounds when combined with Lemma \ref{Lemma: normal}, Lemma \ref{Lemma: non-normal}, and the concentration property (\ref{non-gauss:alpha}).

\begin{corollary} Consider the model described in (\ref{Def:Reg})-(\ref{Def:Normalize}). Suppose further that $F_0=\Phi$, $\lambda$ is chosen according to the exact option in (\ref{normal}), $p/\alpha>8$, and $n>4\log(2/\alpha)$. Let $c>1$, $\cc=(c+1)/(c-1)$, $\nu= \bsqrt{1+2/\log(2p/\alpha)}/[1-2\bsqrt{\log(2/\alpha)/n}]$, and for any $\ell$ such that $1 < \ell < \bsqrt{n/\log(1/\alpha)}$, set
$\omega^2 =1 + \ell\{\log(1/\alpha)/n\}^{1/2}+\ell^2\log(1/\alpha)/(2n)$ and $\gamma = \alpha^{\ell^2/4}$. If $s \log p$ is relatively small
as compared to $n$, namely  $c\nu\bsqrt{2s\log(2p/\alpha)} \leq n^{1/2}\kappa_{\cc} \rho$ for some $\rho<1$, then with probability at least $1-\alpha - \gamma$,
$$
 \tilde \kappa_{\cc} \|\hat \beta-\beta_0 \|_2 \leq \|\hat \beta-\beta_0\|_{2,n}  \leq B_n \sigma \bbsqrt{\frac{2s\log(2p/\alpha)}{ n}},  \ \ B_n =\frac{ 2( 1 + c ) \nu\omega}{\kappa_{\cc}(1-\rho^2)}.$$
\end{corollary}

\begin{corollary}Consider the model described in (\ref{Def:Reg})-(\ref{Def:Normalize}) and suppose that $F_0=\Phi$, Conditions RE and G hold, and $ (s/n) \log (p/\alpha) \to 0$, as $n \to \infty$. Let $\lambda$ be specified according to either
 the exact or asymptotic option in (\ref{normal}).  There is an $o(1)$ term such that
 with probability at least $1-\alpha-o(\alpha)$,
$$
\kappa \|\hat \beta - \beta_0\|_2  \leq  \|\hat \beta - \beta_0\|_{2,n} \leq  C_n \sigma \bbsqrt{\frac{2s\log (2p/\alpha)}{n}},  \ \ C_n = \frac{2( 1+c )}{\kappa \{1- o(1)\}}.
$$
\end{corollary}

\begin{corollary}
Consider the model described in (\ref{Def:Reg})-(\ref{Def:Normalize}). Suppose that Conditions RE, M, and R hold, and $ (s/n) \log (p/\alpha) \to 0$ as $n \to \infty$. Let $\lambda$ be specified according to the asymptotic, exact,
or semi-exact option in (\ref{non-normal}).  There is an $o(1)$ term such that
 with probability at least $1-\alpha-o(\alpha)$
$$
\kappa \|\hat \beta - \beta_0\|_2  \leq  \|\hat \beta - \beta_0\|_{2,n} \leq  C_n \sigma \bbsqrt{\frac{2s\log (2p/\alpha)}{n}},  \ \ C_n = \frac{2( 1+c )}{\kappa \{1- o(1)\}}.
$$
\end{corollary}
As in Lemma \ref{Lemma: non-normal}, in order to achieve Gaussian-like asymptotic conclusions in the non-Gaussian case, we impose stronger restrictions on the growth of $p$ relative to $n$.

\section{Computational properties of $\LASSO$}\label{Sec:Comp}

The second main result of this paper is to formulate the $\LASSO$
as a conic programming problem, with constraints given by a
second-order cone, also informally known as the ice-cream cone.
This allows us to implement the estimator using
efficient algorithmic methods, such as
interior-point methods, which provide polynomial-time
bounds on computational time \cite{NesterovNemirovski,RenegarBook}, and modern first-order methods that have
been recently extended to handle very large conic programming problems
\cite{Nesterov2005,Nesterov2007,LanLuMonteiro,BeckTeboulle2009}.   Before describing the details,
it is useful to recall that a conic programming problem takes the form $\min_u c'u$
subject to $Au =b$ and $u \in C$, where $C$ is a cone. Conic programming  has a tractable dual form
$\max_w b'w$ subject to $w'A + s=c$ and $s \in C^*$, where $C^*=\{ s: s'u \geq 0, \text{ for all } u \in C\}$ is
the dual cone of $C$. A particularly important, highly tractable
class of problems arises when $C$ is the ice-cream cone, $C =
Q^{n+1} = \{ (v,t) \in \Bbb{R}^{n} \times \Bbb{R}: t \geq
\|v\|\}$, which is self-dual, $C=C^*$.

The $\LASSO$ optimization problem is precisely a conic programming problem with second-order conic constraints. Indeed, we can reformulate (\ref{Def:LASSOmod}) as follows:
 \begin{equation}\label{Def:LASSOmod2}
 \displaystyle \min_{t,v,\beta^+,\beta^-} \ \displaystyle  \frac{t}{ n^{1/2}} +
 \frac{\lambda}{n} \sum_{i=1}^p \left(\beta_j^+ + \beta_j^-\right) \
 :  \begin{array}{rl}
&\displaystyle  v_i = y_i - x_i'\beta^+ + x_i'\beta^-, \ i=1,\ldots,n,  \\
&(v,t) \in Q^{n+1}, \ \beta^+ \in \Bbb{R}^p_+, \ \beta^- \in \Bbb{R}^p_+. \\ \end{array}
\end{equation}
Furthermore, we can show that this problem admits the following strongly dual problem: \begin{equation}\label{Def:DualL1}
 \max_{a \in \Bbb{R}^n}   \frac{1}{n} \sum_{i=1}^n  y_i a_i \ \
: \ \  \begin{array}{rl}
&  |\sum_{i=1}^n x_{ij}a_i/n|  \leq \lambda/n, \ j =1,\ldots,p,   \|a \|\leq  n^{1/2}.
\end{array}
\end{equation}
Recall that strong duality holds between a primal and its dual problem if their optimal values are the same, i.e., there is no duality gap.   This is typically an assumption needed for interior-point methods and first-order methods to work. From a statistical perspective, this dual problem maximizes the sample correlation of the score variable $a_i$ with the outcome variables $y_i$
subject to the constraint that the score $a_i$ is approximately
uncorrelated with the covariates $x_{ij}$. The optimal scores $\hat a_i$ equal the residuals $y_i -x_i'\hat \beta$, for all $i=1,\ldots,n$, up to a renormalization factor; they play a key role in deriving sparsity bounds on $\hat \beta$. We formalize the preceding discussion in the following theorem.
\begin{theorem}\label{theorem:computational} The $\LASSO$ problem (\ref{Def:LASSOmod}) is equivalent to the conic programming problem (\ref{Def:LASSOmod2}), which admits the strongly dual problem (\ref{Def:DualL1}). Moreover, if the solution $\hat \beta$ to the problem (\ref{Def:LASSOmod}) satisfies $Y \neq X\hat\beta$, the solution $\hat \beta^+$,  $\hat \beta^{-}$, $\hat v = (\hat v_1,\ldots,\hat v_n)$ to (\ref{Def:LASSOmod2}), and the solution $\hat a$ to  (\ref{Def:DualL1})  are related via  $\hat \beta = \hat \beta^{+} -\hat \beta^{-}$, $\hat v_i = y_i -x_i'\hat \beta$ $(i=1,\ldots,n)$, and $\hat a =  n^{1/2} \hat v/\|\hat v\| $.
\end{theorem}

The conic formulation and the strong duality demonstrated in Theorem 2 allow us to employ both the interior-point and
first-order methods for conic programs to compute the square-root lasso.  We have implemented
both of these methods, as well as a coordinatewise method, for the square-root lasso and made the code
available through the authors' webpages. The square-root lasso
runs at least as fast as the corresponding implementations of these methods for the lasso, for instance,
the Sdpt3 implementation of interior-point method \cite{SDPT3:2010},
and the Tfocs implementation of first-order methods by Becker, Cand\`{e}s and Grant described in \cite{BeckerCandesGrant}.
%the 2010 arXiv working paper ``Templates for Convex Cone Problems with Applications to Sparse Signal Recovery." %
We report the exact running times in the Supplementary Material.

\section{Empirical performance of $\LASSO$ relative to lasso}\label{Sec:Empirical}

In this section we use Monte Carlo experiments to assess the finite sample performance of (i) the infeasible lasso with known $\sigma$ which is unknown outside the experiments, (ii) the post infeasible lasso, which applies ordinary least squares to the model selected by infeasible lasso, (iii) the $\LASSO$ with unknown $\sigma$,  and (iv) the post $\LASSO$, which applies ordinary least squares to the model selected by $\LASSO$.

We set the penalty level for the infeasible lasso and the $\LASSO$ according to the asymptotic options (\ref{Def:LASSOLambda}) and (\ref{our penalty}) respectively, with $1-\alpha=0.95$ and $c=1.1$.
We have also performed experiments where we set the penalty levels according to the exact option. The results are similar, so we do not report them separately.

We use the linear regression model stated in the introduction as a data-generating process, with  either standard normal or $t(4)$ errors:
(a) $\epsilon_i \sim  N(0,1)$,  $(b) \ \epsilon_i  \sim t(4)/2^{1/2}$,
so that $\Ep(\epsilon^2_i) = 1$ in either case. We set the true parameter value as $
 \beta_0 = (1,1,1,1,1,0,\ldots,0)',$ and vary $\sigma$ between $0.25$ and $3$.  The number of regressors  is $p=500$,  the sample size is $n =100$, and we used $1000$ simulations for each design. We generate  regressors as $x_i \sim N(0, \Sigma)$ with the Toeplitz correlation matrix $\Sigma_{jk}=(1/2)^{|j-k|}$. We use as benchmark the performance of the oracle estimator with known true support of $\beta_0$ which is unknown outside the experiment.

%\begin{figure}[!h]
%\centering
%\includegraphics[width=6in, %height=2in]{SQRT_LASSO_RelEMP_RISK_JAN2011bwNEW.eps}%{SQRT_LASSO_EMP_RISKmod2000.eps}
%\color{white}{\figurebox{-0.5pc}{35pc}{}}
%\color{black}{}
%\caption{The average relative empirical risk of the estimators with respect to the oracle estimator, that knows the true support, as a function of the standard deviation of the noise $\sigma$.}\label{Fig:MCfirst01}
%\end{figure}
\begin{figure}[!h]
\centering
\includegraphics[width=\textwidth]%[width=6in, height=2in]
{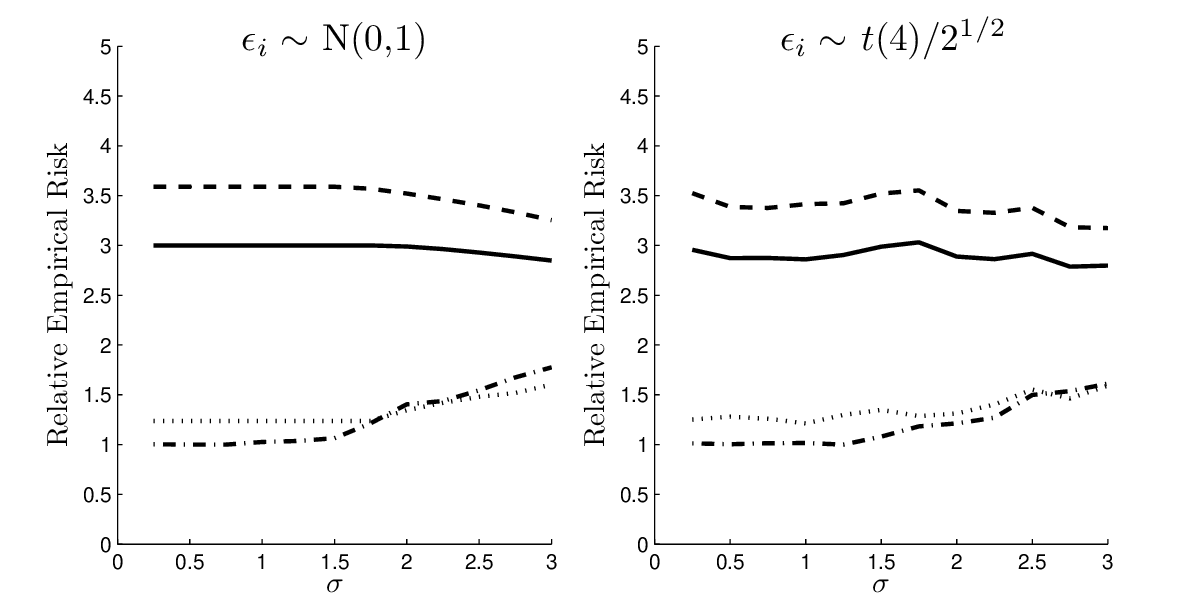}%{SQRT_LASSO_EMP_RISKmod2000.eps}
\caption{Average relative empirical risk of infeasible lasso (solid), square-root lasso (dashes), post infeasible lasso (dots), and post square-root lasso (dot-dash),  with respect to the oracle estimator, that knows the true support, as a function of the standard deviation of the noise $\sigma$.
%A graph showing the truth (dot-dash), an estimate (dashes), another estimate
%(solid), and 95\% pointwise confidence limits (small dashes).
}\label{Fig:MCfirst01}
\end{figure}

\begin{figure}[!h]
\centering
\includegraphics[width=\textwidth]%[width=6in, height=2in]
%[width=0.9\textwidth]
{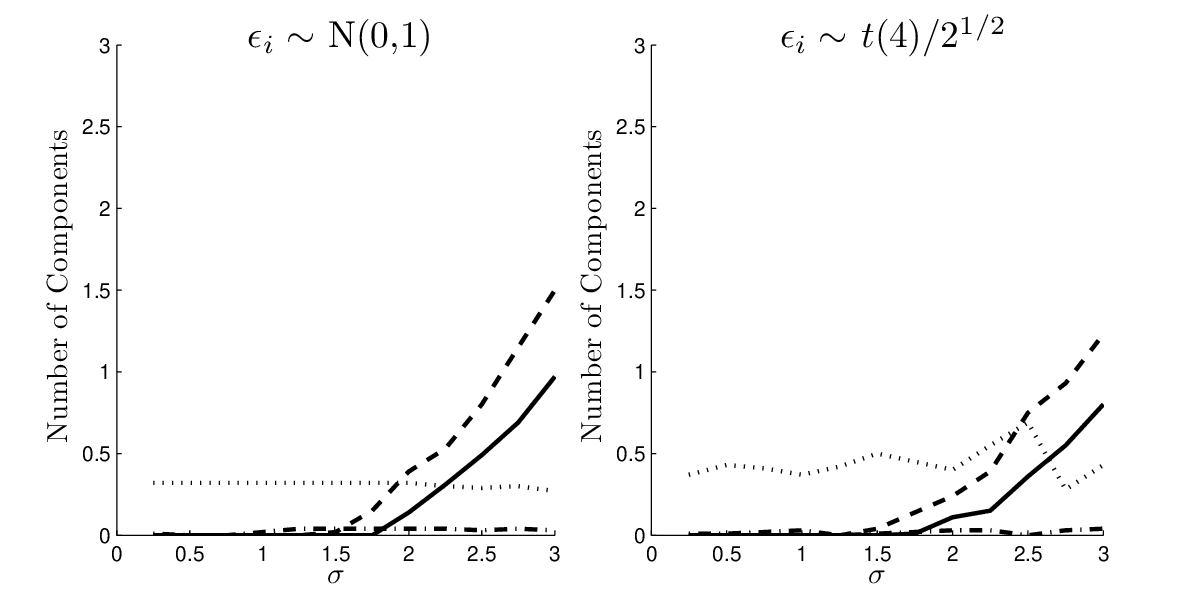}%{SQRT_LASSO_SPARSITYmod2000.eps}
\caption{
%Let $T=\supp(\beta_0)$ denote the true unknown support and $\widehat T$ denote the support selected by the infeasible lasso or square-root lasso estimator. As a function of the noise level $\sigma$, the graph displays the average number of regressors missed from the true support $\Ep(|T\setminus \widehat T|)$ for infeasible lasso (solid) and square-root lasso (dashes), and the average number of regressors selected outside the true support $\Ep(|\widehat T\setminus T|)$ for infeasible lasso (dots) and square-root lasso (dot-dash).
Average number of regressors missed from the true support for infeasible lasso (solid) and square-root lasso (dashes), and the average number of regressors selected outside the true support for infeasible lasso (dots) and square-root lasso (dot-dash), as a function of the noise level $\sigma$.
}\label{Fig:MCfirst03}
\end{figure}

We present the results of computational experiments for designs (a)
and (b) in Figs. \ref{Fig:MCfirst01} and \ref{Fig:MCfirst03}. For each model,
Figure \ref{Fig:MCfirst01} shows the relative average empirical risk with respect to the oracle estimator $\beta^*$,  $\Ep(\|\widetilde \beta-\beta_0\|_{2,n})  / \Ep(\|\beta^*-\beta_0\|_{2,n})$, % $\Ep[\bsqrt{\En\{(x'(\widetilde \beta-\beta_0))^2\}}]/\Ep[\bsqrt{\En[(x'( \beta^*-\beta_0))^2]}];$
 and Figure \ref{Fig:MCfirst03} shows the average number of regressors missed from the true model and the average number of regressors selected outside the true model, $\Ep(|\supp(\beta_0)\setminus\supp(\widetilde \beta)|)$  and $\Ep(|\supp(\widetilde \beta)\setminus\supp(\beta_0)|)$,  respectively.

Figure \ref{Fig:MCfirst01} shows the empirical risk of the estimators. We see that, for a wide range of the noise level  $\sigma$, the $\LASSO$ with unknown $\sigma$ performs comparably to the infeasible lasso with known $\sigma$.   These results agree with our theoretical results, which state that the upper bounds on empirical risk for the $\LASSO$ asymptotically approach the analogous bounds for infeasible lasso. The finite-sample differences in empirical risk for the infeasible lasso and the $\LASSO$ arise primarily due to the $\LASSO$ having a larger bias than the infeasible lasso.  This bias arises because the $\LASSO$ uses an effectively heavier penalty induced by $\widehat Q(\hat\beta)$ in place of $\sigma^2$; indeed, in these experiments, the average values of  $\widehat Q(\hat\beta)^{1/2}/\sigma$ varied between $1.18$ and $1.22$.

Figure \ref{Fig:MCfirst01} shows that the post $\LASSO$  substantially outperforms both the infeasible lasso and the $\LASSO$.  Moreover, for a wide range of $\sigma$, the post $\LASSO$  outperforms the post infeasible lasso.   The post square-root lasso is able to improve over the $\LASSO$ due to removal of the relatively large shrinkage bias of the $\LASSO$.  Moreover,  the post $\LASSO$ is able to outperform the post infeasible lasso  primarily due to its better sparsity properties, which can be observed from Figure \ref{Fig:MCfirst03}. These results on the post $\LASSO$ agree closely with our theoretical results reported in the arXiv working paper ``Pivotal Estimation of Nonparametric Functions via Square-root Lasso" by the authors,
which state that the  upper bounds on empirical risk for the post $\LASSO$  asymptotically are no larger than the analogous bounds for the $\LASSO$ or the infeasible lasso, and can be strictly better when the $\LASSO$ acts as a near-perfect model selection device. We see this happening in Figure \ref{Fig:MCfirst01}, where as the noise level $\sigma$ decreases, the post $\LASSO$ starts to perform as well as the oracle estimator. As we see from Figure \ref{Fig:MCfirst03}, this happens because as $\sigma$ decreases, the $\LASSO$ starts to select the true model nearly perfectly, and hence the post $\LASSO$ starts to become the oracle estimator with a high probability.

Next let us now comment on the difference between the normal and $t(4)$ noise cases, i.e., between the right and left panels in Figure \ref{Fig:MCfirst01} and \ref{Fig:MCfirst03}. We see that the results for the Gaussian case carry over to $t(4)$ case with nearly undetectable changes. In fact, the performance of the infeasible lasso and the $\LASSO$ under $t(4)$ errors nearly coincides with their performance under Gaussian errors, as predicted by our theoretical results in the main text, using moderate deviation theory, and in the Supplementary Material, using self-normalized moderate deviation theory.

In the Supplementary Material, we provide further Monte Carlo comparisons that include asymmetric error distributions, highly correlated designs, and feasible lasso estimators based on the use of conservative bounds on $\sigma$ and cross validation.  Let us briefly summarize the key conclusions from these experiments.  First, presence of asymmetry in the noise distribution and of a high correlation in the design does not change the results qualitatively. Second, naive use of conservative bounds on $\sigma$ does not result in good feasible lasso estimators. Third, the use of cross validation for choosing the penalty level does produce a feasible lasso estimator performing well in terms of empirical risk but poorly in terms of model selection.  Nevertheless, even in terms of empirical risk, the cross-validated lasso is  outperformed by the post square-root lasso.  In addition, cross-validated lasso is outperformed by the square-root-lasso with the penalty level scaled by 1/2. This is noteworthy, since the estimators based on the square-root lasso are much cheaper computationally.  Lastly, in the 2011 arXiv working paper ``Pivotal Estimation of Nonparametric Functions via Square-Root Lasso" we provide a further analysis of the post square-root lasso estimator and generalize the setting of the present paper to the fully nonparametric regression model.

\section*{Acknowledgement}
%We would like to thank Isaiah Andrews, Denis Chetverikov, Joseph Doyle, Ilze %Kalnina, and James MacKinnon, the editor Anthony Davison, and two referees for %extremely helpful comments. The research of Alexandre Belloni and Victor %Chernozhukov was supported by the NSF grant SES-0752266;
%and the research of Lie Wang was supported by NSF Grant DMS-1005539.
We would like to thank Isaiah Andrews, Denis Chetverikov, Joseph Doyle, Ilze Kalnina, and James MacKinnon for very helpful comments. The authors also thank the editor and two referees for suggestions that greatly improved the paper. This work was supported by the National Science Foundation.

\section*{Supplementary Material}

The online Supplementary Material contains a complementary analysis of the penalty choice based on moderate deviation theory for self-normalizing sums, discussion on computational aspects of $\LASSO$ as compared to lasso, and additional Monte Carlo experiments.  We also provide the omitted part of proof of Lemma 2, and list the inequalities used in the proofs.

\appendix

\section*{Appendix 1}

\subsection*{Proofs of Theorems 1 and 2}

\begin{proof}[Proof of Theorem
\ref{Thm:RateSquareRootLASSONonparametric}]  Step 1. We show that $\hat \delta = \hat \beta - \beta_0 \in \Delta_{\cc}$ under the prescribed penalty level.  By definition of
$\widehat \beta$
\begin{equation}\label{Def:OPT}
\bsqrt{\widehat Q(\hat \beta)} - \bsqrt{\widehat Q (\beta_0)} \leq  \frac{\lambda}{n} \| \beta_0\|_{1} - \frac{\lambda}{n}\|\hat \beta\|_{1} \leq \frac{\lambda}{n} (\| \hat \delta_T\|_{1} - \| \hat \delta_{T^c}\|_{1}),
\end{equation}
where the last inequality holds because
 \begin{eqnarray}\label{Rel:(b)2}
 \|\beta_0\|_{1} - \| \hat \beta\|_{1} &= &  \| \beta_{0T}\|_{1} - \|\hat \beta_T\|_{1} - \| \hat \beta_{T^c} \|_{1}
  \leqslant   \| \hat \delta_T\|_{1} - \| \hat \delta_{T^c}\|_{1}.
 \end{eqnarray}
Also, if $\lambda \geqslant cn\|\tilde S\|_{\infty}$ then
 \begin{eqnarray}\label{Rel:(aa)2}
 \bsqrt{\hat Q (\hat \beta)} - \bsqrt{\hat Q(\beta_0)}
  & \geqslant &
   \tilde S'  \hat \delta
  \geqslant
  - \|\tilde S\|_{\infty} \| \hat \delta\|_{1}
  \geqslant   - \frac{\lambda}{cn} (\| \hat \delta_T\|_{1} + \| \hat \delta_{T^c}\|_{1}),
 \end{eqnarray}
 where the first inequality hold by convexity of $\hat Q^{1/2}$.
Combining (\ref{Def:OPT}) with (\ref{Rel:(aa)2}) we obtain
\begin{equation}\label{keykey}
-\frac{\lambda}{cn} (\| \hat \delta_T\|_{1} + \| \hat \delta_{T^c}\|_{1}  )  \leqslant \frac{\lambda}{n} (
\|\hat \delta_T\|_{1} - \|\hat \delta_{T^c}\|_{1}),
\end{equation}
that is
 \begin{equation}\label{domination3}
\|\hat \delta_{T^c}\|_{1} \leqslant \frac{c+1}{c-1}   \| \hat \delta_{T}\|_{1} = \cc \| \hat \delta_{T}\|_{1}.
 \end{equation}

Step 2. We derive bounds on the estimation error.  We shall use the following relations: \begin{eqnarray}
& \hat Q (\hat \beta) - \hat Q(\beta_0)   & =     \|\widehat  \delta\|^2_{2,n} - 2 \En(\sigma \epsilon x'\widehat \delta), \label{Rel:trivial}\\
 & \widehat Q(\widehat \beta) - \widehat Q(\beta_0)  & = \left[\bsqrt{\widehat Q(\widehat \beta)} + \bsqrt{\widehat Q(\beta_0)}\right]\left[\bsqrt{\widehat Q(\widehat \beta)} - \bsqrt{\widehat Q(\beta_0)}\right],  \label{bn:Id} \\
\label{Rel:(a)2}
& 2 |\En(\sigma \epsilon x'\widehat \delta)| &   \leq   2  \bsqrt{\hat Q (\beta_0)}\|\tilde S\|_{\infty} \| \hat \delta\|_1, \label{rel-bound} \\
& \|\hat\delta_T\|_{1} &    \leq  \frac{ s^{1/2}\|\widehat\delta\|_{2,n}}{\kappa_{\cc}} ,  \hat \delta \in \Delta_{\cc},
\label{rel-step-1}
 \end{eqnarray}
 where (\ref{Rel:(a)2}) holds by Holder inequality and (\ref{rel-step-1})  holds by the definition of $\kappa_{\bar c}$.

 Using (\ref{Def:OPT}) and (\ref{Rel:trivial})--(\ref{rel-step-1}) we obtain
{\small \begin{equation}\label{Eq:Id}\|\widehat \delta\|_{2,n}^2  \leq   2\bsqrt{\widehat Q(\beta_0)} \|\tilde S\|_{\infty}\|\hat\delta\|_{1}  + \left[\bsqrt{\widehat Q(\widehat \beta)} + \bsqrt{\widehat Q(\beta_0)}\right] \frac{\lambda}{n} \Big( \frac{s^{1/2}\|\widehat\delta\|_{2,n}}{\kappa_{\cc}} - \|\hat\delta_{T^c}\|_{1}\Big).\end{equation}}
Also using (\ref{Def:OPT}) and (\ref{rel-step-1}) we obtain
 \begin{eqnarray}\label{Eq:boundQ}
\bsqrt{\widehat Q(\hat \beta)}
 \leq   \bsqrt{\widehat Q (\beta_0)}  + \frac{\lambda}{n} \left(  \frac{ s^{1/2}\|\widehat\delta\|_{2,n}}{\kappa_{\cc}}\).
 %- \|\hat\delta_{T^c}\|_{1}\right).
 \end{eqnarray}
Combining inequalities (\ref{Eq:boundQ}) and (\ref{Eq:Id}), we obtain
$\|\hat\delta\|_{2,n}^2  \leq   2\bsqrt{\widehat Q(\beta_0)} \|\tilde S\|_{\infty}\|\hat\delta\|_{1}  + 2\bsqrt{\widehat Q(\beta_0)} \frac{\lambda s^{1/2}}{n\kappa_{\cc}}  \|\hat \delta\|_{2,n}   +   \(\frac{\lambda s^{1/2}}{n\kappa_{\cc}}\|\hat \delta\|_{2,n}\)^2- 2\bsqrt{\widehat Q(\beta_0)} \frac{\lambda}{n}\|\hat\delta_{T^c}\|_{1}.$
Since $\lambda \geqslant cn\|\tilde S\|_{\infty}$, we obtain
 \begin{eqnarray*}
\|\hat\delta\|_{2,n}^2 & \leq &  2\bsqrt{\widehat Q(\beta_0)} \|\tilde S\|_{\infty}\|\hat\delta_T\|_{1}  + 2\bsqrt{\widehat Q(\beta_0)} \frac{\lambda s^{1/2}}{n\kappa_{\cc}}  \|\hat \delta\|_{2,n}    +   \(\frac{\lambda s^{1/2}}{n\kappa_{\cc}}\|\hat \delta\|_{2,n}\)^2,
 \end{eqnarray*}
and then using (\ref{rel-step-1}) we obtain
 \begin{eqnarray*}
\left\{1 - \Big(\frac{\lambda s^{1/2}}{n\kappa_{\cc}}\Big)^2 \right\}\|\hat\delta\|_{2,n}^2 & \leq &  2 \left( \frac{1}{c} +1 \right) \bsqrt{\widehat Q(\beta_0)} \frac{\lambda s^{1/2}}{n\kappa_{\cc}}  \|\hat \delta\|_{2,n}.
 \end{eqnarray*}
 Provided that $(n\kappa_{\cc})^{-1}\lambda s^{1/2} \leq \rho < 1$ and solving the inequality above
 we obtain the bound stated in the theorem. \end{proof}

% which holds by assumption, we obtain
%$$
%\|\hat \delta\|_{2,n}  \leq  4\left( \frac{1}{c} +1 \right)\bsqrt{\widehat Q(\beta_0)} \frac{\lambda s^{1/2}}{n\kappa_{\cc}},
%$$
%which is equal to the bound stated in the theorem.

\begin{proof}[Proof of Theorem \ref{theorem:computational}] The equivalence
of $\LASSO$ problem (\ref{Def:LASSOmod}) and the conic programming problem (\ref{Def:LASSOmod2}) follows
immediately from the definitions.  To establish the duality,
  for $e=(1,\ldots,1)'$, we can write (\ref{Def:LASSOmod2}) in  matrix form as
  $$
 \min_{t,v,\beta^+,\beta^-}  \frac{t}{ n^{1/2}} + \frac{\lambda}{n}e' \beta^+
+ \frac{\lambda}{n} e' \beta^-  \ : \ \begin{array}{lll}
& v + X\beta^+ - X\beta^- = Y \\
& (v,t) \in Q^{n+1}, \ \beta^+ \in \Bbb{R}^p_+, \ \beta^- \in \Bbb{R}^p_+. \\
\end{array}
$$
By the conic duality theorem, this has dual
$$
 \max_{a,s^t,s^v,s^+,s^-}  Y'a  \ : \
\begin{array}{rl}
& s^t = 1/ n^{1/2},  a + s^v = 0, X'a + s^+ = \lambda e/n,  -X'a + s^- = \lambda  e/n\\
& (s^v,s^t) \in Q^{n+1}, \ s^+ \in \Bbb{R}^p_+, \ s^- \in \Bbb{R}^p_+.
\end{array}
$$
The constraints $X'a + s^+ = \lambda/n$ and $-X'a + s^- = \lambda/n$ leads to $\|X'a\|_\infty \leq \lambda/n$.
The conic constraint $(s^v,s^t) \in Q^{n+1}$ leads to $1/ n^{1/2}=s^t \geq \|s^v\| = \|a\|$. By scaling the variable $a$ by $n$ we obtain the stated dual problem.

Since the primal problem is strongly feasible, strong duality holds by Theorem 3.2.6 of \cite{RenegarBook}. Thus, by strong duality, we have
$ n^{-1}\sum_{i=1}^ n  y_i \hat a_i = n^{-1/2}\|Y - X\hat\beta\| + n^{-1}\lambda\sum_{j=1}^p |\hat\beta_j|.$
Since $n^{-1}\sum_{i=1}^ n  x_{ij}\hat a_i \hat\beta_j =
\lambda|\hat\beta_j|/n$ for every $j=1,\ldots,p$, we
have
$$ \frac{1}{n}\sum_{i=1}^ n  y_i \hat a_i = \frac{\|Y - X\hat\beta\|}{ n^{1/2}} + \sum_{j=1}^p  \frac{1}{n}\sum_{i=1}^ n  x_{ij}\hat a_i \hat\beta_j= \frac{\|Y - X\hat\beta\|}{ n^{1/2}} + \frac{1}{n}\sum_{i=1}^ n \hat a_i \sum_{j=1}^p x_{ij}\hat\beta_j.$$
Rearranging the terms we have $n^{-1}\sum_{i=1}^n\{ (y_i-x_i'\hat\beta) \hat
a_i\} = \|Y - X\hat\beta\|/ n^{1/2}$. If $\|Y-X\hat\beta\|>0$, since $\|\hat a\|\leq
 n^{1/2}$, the equality can only hold for $\hat a =
 n^{1/2}(Y-X\hat\beta)/\|Y-X\hat\beta\|=(Y-X\hat\beta)/\bsqrt{\hat
Q(\hat\beta)}$. \end{proof}

\section*{Appendix 2}

\subsection*{Proofs of Lemmas 1 and 2}

\begin{proof}[Proof of Lemma \ref{Lemma: normal}]
Statement (i) holds by definition. To show
statement (ii), we define $t_n = \Phi^{-1}(1- \alpha/2p)$ and $0<r_n=\ell\bsqrt{\log(1/\alpha)/n}<1$. It is
known that $\log(p/\alpha)<t_n^2<2\log(2p/\alpha)$ when
$p/\alpha>8$. Then since $Z_j = n^{1/2} \En (x_{j} \epsilon) \sim N(0,1)$ for each $j$, conditional
on $X$, we have
by the union bound and $ F_0 = \Phi$,
$
 \Pr( c \Lambda > c n^{1/2} t_n \mid X )  \leq   p \ \Pr\{|Z_j| > t_n (1-r_n) \mid X\} + \Pr \{ \En(\epsilon^2) < (1 - r_n)^2 \} \leq   2p\ \bar \Phi\{t_n (1- r_n) \}+\Pr\{ \En(\epsilon^2)
< (1 - r_n)^2 \}.$ Statement (ii) follows by observing that by Chernoff tail bound for $\chi^2(n)$, Lemma 1 in \cite{LaurentMassart2000},
$\Pr \{ \En(\epsilon^2) < (1 - r_n)^2\} \leq \exp(-nr_n^2/4)$, and
 \begin{eqnarray*}
 2p\ \bar \Phi\{t_n (1- r_n) \} && \leq
2p\frac{\phi\{t_n(1-r_n)\}}{t_n(1-r_n)}=2p\frac{\phi(t_n)}{t_n}
\frac{\exp(t_n^2r_n-\frac{1}{2}t_n^2r_n^2)}{1-r_n}\\
& & \leq 2p \bar \Phi(t_n)\frac{1+t_n^2}{t_n^2}
\frac{\exp(t_n^2r_n-\frac{1}{2}t_n^2r_n^2)}{1-r_n}\leq \alpha
\left(1+\frac{1}{t_n^2}\right)\frac{\exp(t_n^2r_n)}{1-r_n} \\
& & \leq
\alpha\left\{1+\frac{1}{\log(p/\alpha)}\right\}\frac{\exp\{2\log(2p/\alpha)r_n\}}{1-r_n},
\end{eqnarray*}
where we have used the inequality $\phi(t) t/(1 + t^2) \leq \bar \Phi(t) \leq \phi(t)/t$ for $t > 0$.

For statement (iii), it is sufficient to show that $\Pr(\Lambda_\Phi >
 v  n^{1/2} t_n \mid X )\leq\alpha$. It can be seen that
there exists a $v'$ such that $v'>\bsqrt{1+2/\log(2p/\alpha)}$ and
$1-v'/v>2\bsqrt{\log(2/\alpha)/n}$ so that
$
 \Pr(\Lambda_\Phi >  v n^{1/2} t_n \mid X  )  \leq   p \max_{1 \leq j \leq p} \Pr( |Z_j| > v't_n \mid X)+\Pr\{\En(\epsilon^2)<(v'/v)^2 \}
  = 2p\bar
\Phi(v't_n)+\Pr[\bsqrt{\En(\epsilon^2)}<v'/v].
$
Proceeding as before, by Chernoff tail bound for $\chi^2(n)$, $\Pr[\bsqrt{\En(\epsilon^2)}<v'/v]\leq  \exp\{-n(1-v'/v)^2/4\}\leq \alpha/2,$ and
\begin{eqnarray*}
2p\bar \Phi(v't_n)&\leq &
2p\frac{\phi(v't_n)}{v't_n}=2p\frac{\phi(t_n)}{t_n}\frac{\exp\{-\frac{1}{2}t_n^2(v'^2-1)\}}{v'}\\
&\leq& 2p\bar\Phi(t_n)
\left(1+\frac{1}{t_n^2}\right)\frac{\exp\{-\frac{1}{2}t_n^2(v'^2-1)\}}{v'}\\
&=&\alpha\left(1+\frac{1}{t_n^2}\right)\frac{\exp\{-\frac{1}{2}t_n^2(v'^2-1)\}}{v'}\\
&\leq&
\alpha\left\{1+\frac{1}{\log(p/\alpha)}\right\}\frac{\exp\{-\log(2p/\alpha)(v'^2-1)\}}{v'}\\
&\leq& 2\alpha \exp\{-\log(2p/\alpha)(v'^2-1)\}<\alpha/2.
\end{eqnarray*}
Putting the inequalities together, we conclude that $\Pr(\Lambda_\Phi > v  n^{1/2}
t_n \mid X )\leq\alpha$.

Finally, the asymptotic result follows directly from the finite sample bounds and noting that $p/\alpha\to \infty$ and that under the growth condition we can choose $\ell\to \infty$ so that $\ell\log(p/\alpha)\log^{1/2}(1/\alpha)=o( n^{1/2})$.
\end{proof}

\begin{proof}[Proof of Lemma \ref{Lemma: non-normal}]  Statements (i) and (ii) hold by definition.
To show (iii), consider first the case of $2 < q  \leq 8$, and define $t_n = \Phi^{-1}(1- \alpha/2p)$ and
$r_n = \alpha^{-\frac{2}{q}}  n^{-\{(1-2/q) \wedge1/2\}} \ell_n, $
for  some $\ell_n$ which grows to infinity but so slowly that the condition stated below is satisfied.
Then for any $F_0=F_{0n}$ and $X=X_n$ that obey Condition M:
\begin{eqnarray*}
 && \Pr ( c \Lambda > c  n^{1/2} t_n \mid X  )  \\
&& \leq_{(1)}   p \max_{1 \leq j \leq p} \Pr\{ | n^{1/2} \En (x_{j} \epsilon) | > t_n (1-r_n) \mid X\} + \Pr[ \bsqrt{\En(\epsilon^2)} < 1 - r_n ] \\
& & \leq_{(2)}   p \max_{1 \leq j \leq p} \Pr\{ | n^{1/2} \En (x_{j} \epsilon) | > t_n (1-r_n) \mid X \} + o(\alpha) \\
& &  =_{(3)}  2p \ \bar \Phi\{t_n (1- r_n)\}\{1+ o(1)\} + o(\alpha) \\
& &  =_{(4)}   2p \ \frac{\phi\{t_n (1- r_n)\}}{t_n (1- r_n)} \{1+o(1)\} + o(\alpha) \label{r3}\\
& & =   2p \  \frac{\phi(t_n)}{t_n}  \frac{\exp(t_n^2r_n-t_n^2r_n^2/2)}{1-r_n}\{1+o(1)\} + o(\alpha)  \nonumber \\
& & =_{(5)}  2p \ \frac{\phi(t_n)}{t_n}  \{1+o(1)\} + o(\alpha) \label{r4}  =_{(6)}  2p \ \bar \Phi(t_n)\{1+o(1)\} + o(\alpha) = \alpha \{1+o(1)\} \label{r5},
\end{eqnarray*}
where (1) holds by the union bound; (2) holds by the application of either Rosenthal's inequality (Rosenthal 1970) for the case of
$q >4$ and Vonbahr--Esseen's inequalities (von Bahr \& Esseen 1965) for the case of $2 < q \leq 4$,
\begin{equation}\label{Eq:PRalpha1}
\Pr [ \bsqrt{\En(\epsilon^2)} < 1 - r_n ] \leq \Pr \{ |\En(\epsilon^2) -1| > r_n \} \lesssim  \alpha \ell_n^{-q/2} = o(\alpha),
\end{equation}
 (4) and (6) by  $\phi(t)/t \sim \bar \Phi(t)$ as $t \to \infty$;
(5) by   $t_n^2r_n =  o(1)$, which holds if
$ \log(p/\alpha) \alpha^{-\frac{2}{q}}  n^{-\{(1-2/q) \wedge1/2\}} \ell_n$ $=$ $o(1).$
Under our condition $\log(p/\alpha) = O (\log n)$, this condition is satisfied for some slowly growing
$\ell_n$, if
\begin{equation}\label{bound alpha}
\alpha^{-1} = o\{ n^{ (q/2-1) \wedge q/4}/\log^{q/2} n\}.
\end{equation}
 To verify relation (3), by Condition M and Slastnikov's theorem on moderate deviations, see \cite{slastnikov} and \cite{rubin:sethuraman},
we have that uniformly
in $ 0 \leq |t| \leq  k \log^{1/2} n$ for some $k^2 < q -2$, uniformly in $1 \leq j \leq p$ and for any $F_0= F_{0n} \in \mathcal{F}$,
$
\Pr\{ n^{1/2} |\En (x_{j} \epsilon)| > t\mid X\}/\{2\bar \Phi(t)\} \to 1,
$
so the relation (3) holds for $t = t_n (1-r_n) \leq  \bsqrt{ 2 \log (2p/\alpha) } \leq \bsqrt{\eta (q-2) \log n}$ for
$\eta<1$ by Condition R.  We apply Slastnikov's theorem  to $n^{-1/2} |\sum_{i=1}^n z_{i,n}|$ for  $z_{i, n} = x_{ij} \epsilon_{i}$, where we allow the design $X$, the law $F_0$, and index $j$ to be  indexed by $n$.
Slastnikov's theorem then applies provided $\sup_{n, j \leq p} \En \{\Ep_{F_0} ( |z_{n}|^q)\}  =
 \sup_{n, j \leq p} \En ( |x_{j}|^q)  \Ep_{F_0} (|\epsilon|^q)< \infty,$ which is implied by our Condition M,
 and where we used the condition that the design is fixed, so that $\epsilon_i$ are independent of $x_{ij}$.
  Thus,  we obtained the  moderate deviation result uniformly in $1\leq j \leq p$ and for any sequence of
 distributions $F_0=F_{0n}$ and designs $X=X_n$ that obey our Condition M.

Next suppose that $q \geq 8$. Then the same argument applies, except
that now relation (2) could also be established by using Slastnikov's
theorem on moderate deviations. In this case redefine $r_n =   k \bsqrt{\log n/n}; $ then, for some constant $k^2< \bsqrt{(q/2)-2}$
we have \begin{equation}\label{Eq:PRalpha2}\Pr \{\En(\epsilon^2) < (1 - r_n)^2 \} \leq \Pr \{ |\En(\epsilon^2)- 1| > r_n \} \lesssim n^{-k^2},\end{equation}
so the relation (2) holds if
\begin{equation}\label{bound alpha 2} 1/\alpha = o(n^{k^2}).
  \end{equation} This applies whenever $q \geq 4$, and this results in weaker requirements on $\alpha$ if $q \geq 8$.
The relation (5) then follows if  $t^2_n r_n =  o(1)$, which is easily satisfied for the new $r_n$, and the result follows.

Combining  conditions in (\ref{bound alpha})
 and (\ref{bound alpha 2}) to give the weakest restrictions on the growth of $\alpha^{-1}$,
 we obtain the growth conditions stated in the lemma.
%\begin{equation}
%\alpha^{-1} = o\left( \begin{array}{lll}  n^{q/2-1}/(\log n)^{q/2} & \text{ if } &  2< q \leq 4 \\
% n^{q/4}/(\log n)^{q/2}  & \text{ if }  &  4 < q \leq 8 \\
% n^{q/2-2} & \text{ if } & 8 < q \\
%\end{array} \right)
%\end{equation}

To show statement (iv) of the lemma, it suffices to show that for any $\nu'>1$, and $F \in \mathcal{F}$, $\Pr (  \Lambda_F >  \nu'  n^{1/2} t_n \mid X )  =  o(\alpha)$,
which follows analogously to the proof of statement (iii); we relegate the details to the Supplementary Material.   \end{proof}

%\bibliographystyle{biometrika}
%\bibliography{paper-ref}
\bibliographystyle{plain}
\bibliography{biblioSqLASSO}

\pagebreak

\vspace{1.0cm}

{\LARGE Supplementary Appendix for ``Square-root lasso: pivotal recovery of sparse signals via conic programming"}

\vspace{1.0cm}

\begin{quote}
\textbf{Abstract.} In this appendix we gather additional theoretical and computational results for ``Square-root lasso: pivotal recovery of sparse signals via conic programming." We include a complementary analysis of the penalty choice based on moderate deviation theory for self-normalizing sums. We provide a
 discussion on computational aspects of $\LASSO$ as compared to lasso. We carry out additional Monte Carlo experiments.  We also provide the omitted part of proof of Lemma 2, and list the inequalities used in the proofs.
\end{quote}

\vspace{1.0cm}

\section{Additional Theoretical Results}

In this section we derive additional results using moderate deviation theory for self-normalizing sums to bound the penalty level. These results are complementary to the results given in the main text since conditions required here are not implied nor imply the conditions there. These conditions require stronger moment assumptions but in exchange they result in weaker growth requirements on $p$ in relation to $n$.

Recall the definition of the choices of penalty levels
\begin{equation}\label{non-normal}
\begin{array}{llll}
&\text{exact:} & \lambda =  c \Lambda_{F_0}(1-\alpha\mid X), \\
%&\text{semi-exact} & \lambda = (1+u_n) \cdot c  \max_{F \in \mathcal{F}}\Lambda(1-\alpha\mid X, F) \\
& \text{asymptotic:} &
\lambda =  c n^{1/2} \Phi^{-1}(1-\alpha/2p),
\end{array}
\end{equation} where $\Lambda_{F_0}(1-\alpha\mid X) = (1-\alpha)$-quantile of $n\|\En(x\epsilon)\|_\infty/\{\En(\epsilon^2)\}^{1/2}$. We will make use of the following condition.

\textsc{Condition SN.} \textit{There is a $q> 4$ such that the noise obeys $\sup_{n\geq 1} \Ep_{F_0} (|\epsilon|^{q}) < \infty,$ and the design $X$ obeys $\sup_{n\geq 1} \max_{1 \leq i \leq n} \|x_{i}\|_\infty  < \infty$. Moreover, we also assume $ \log(p/\alpha) \alpha^{-2/q} n^{-1/2} \log^{1/2} (n\vee p/\alpha) = o(1)$ and $\En(x_{j}^2)=1$ ($j=1,\ldots,p$).}

\begin{lemma}\label{Lemma: non-normalMprimeRprime} Suppose that condition SN holds and $n \to \infty$. Then, (1) the asymptotic option in (\ref{non-normal}) implements $\lambda > c  \Lambda$ with probability at least $1-\alpha \{ 1 + o(1)\}$, and (2)
$$\begin{array}{rl}
 \displaystyle \Lambda_{F_0}\(1-\alpha \ \mid X\) &  \displaystyle \leq \{1 + o(1)\} n^{1/2}\Phi^{-1}(1-\alpha/2p).
 \end{array}$$
\end{lemma}

This lemma in combination with Theorem 1 of the main text imply the following result:

\begin{corollary}
Consider the model described in  the main text. Suppose that Conditions RE and SN hold, and $ (s/n) \log (p/\alpha) \to 0$ as $n \to \infty$. Let $\lambda$ be specified according to the asymptotic or exact option in (\ref{non-normal}).  There is an $o(1)$ term such that
 with probability at least $1-\alpha\{1+o(1)\}$
$$
\kappa \|\hat \beta - \beta_0\|_2  \leq  \|\hat \beta - \beta_0\|_{2,n} \leq  C_n \sigma \left\{\frac{2s\log (2p/\alpha)}{n}\right\}^{1/2},  \ \ C_n = \frac{2( 1+c )}{\kappa \{1- o(1)\}}.
$$
\end{corollary}

\section{Additional Computational Results}

\subsection{Overview of Additional Computational Results}
In the main text we formulated the $\LASSO$ as a convex conic programming problem. This fact allows us to use conic programming methods to compute the $\LASSO$ estimator. In this section we provide further details on these methods, specifically on (i) the first-order methods, (ii) the interior-point methods, and (iii) the componentwise search methods, as specifically adapted to solving conic programming problems.  We shall also compare the adaptation of these methods to $\LASSO$ with the respective adaptation of these methods to lasso.

\subsection{Computational Times} Our main message here is that the average running times for solving lasso and the $\LASSO$ are comparable in practical problems.
We document this in Table \ref{Table:Comp}, where we record the average computational times, in seconds, of the three computation methods mentioned above. The design for computational experiments is the same as in the main text.  In fact, we see that the $\LASSO$ is often slightly easier to compute than the lasso. The table also reinforces the typical behavior of the three principal computational methods. As the size of the optimization problem increases, the running time for an interior-point method grows faster than that for a first-order method. We also see, perhaps more surprisingly, that a simple componentwise method is particularly effective, and this might be due to a high sparsity of the solutions in our examples. An important remark here is that we did not attempt to compare rigorously across different computational methods to isolate the best ones, since these methods have different initialization and stopping criteria and the results could be affected by that. Rather our focus here is comparing the performance of each computational method as applied to lasso and the $\LASSO$. This is an easier comparison problem, since given a computational method, the initialization and stopping criteria are similar for two problems.

\begin{center}
\begin{table}[h]
\begin{center}
\begin{tabular}{rccc}
$n=100$, $p=500$ & Componentwise & First Order & Interior Point\\
\\
lasso &    0$\cdot$2173 &  10$\cdot$99 & 2$\cdot$545 \\
$\LASSO$ &  0$\cdot$3268 &    7$\cdot$345 & 1$\cdot$645\\
\\
$n=200$, $p=1000$ & Componentwise & First Order & Interior Point\\
\\
lasso &    0$\cdot$6115 &  19$\cdot$84 & 14$\cdot$20 \\
$\LASSO$ &  0$\cdot$6448 &    19$\cdot$96 & 8$\cdot$291\\
\\
$n=400$, $p=2000$ & Componentwise & First Order & Interior Point\\
\\
lasso &    2$\cdot$625 &  84$\cdot$12 & 108$\cdot$9 \\
$\LASSO$ &   2$\cdot$687 &   77$\cdot$65 & 62$\cdot$86 \\
\\
\end{tabular}\caption{We use the same design as in the main text, with $s=5$ and $\sigma=1$, we averaged
the computational times over 100 simulations.}\label{Table:Comp}
\end{center}
\end{table}
\end{center}

\subsection{Details on Computational Methods}

Below we discuss in more detail the applications of these methods for the lasso and the $\LASSO$. The similarities between the lasso and the $\LASSO$ formulations derived below provide a theoretical justification for the similar computational performance. \\

{\bf Interior-point methods.} Interior-point method (ipm) solvers typically focus on solving conic programming problems in standard form: \begin{equation}\label{ConicProg} \min_{w} c'w \ : \ Aw = b, w \in K,\end{equation}
 where $K$ is a cone. The main difficulty of the problem arises because the conic constraint is binding at the optimal solution. To overcome the difficulty, ipms regularize the objective function with a barrier function so that the optimal solution of the regularized problem naturally lies in the interior of the cone. By steadily scaling down the barrier function, an ipm creates a sequence of solutions that converges to the solution of the original problem (\ref{ConicProg}).

In order to formulate the optimization problem associated with the lasso estimator as a conic programming problem (\ref{ConicProg}), specifically, associated with the second-order cone $Q^{k+1} = \{ (v,t)\in \RR^{k+1} : t\geq \|v\|\}$, we let $\beta = \beta^+ - \beta^-$ for $\beta^+\geq 0$ and $\beta^-\geq 0$. For any vector  $v\in\RR^n$ and scalar $t\geq 0$, we have that $v'v\leq t$ is equivalent to $\| ( v, (t-1)/2) \|_2 \leq (t+1)/2.$ The latter can be formulated as a second-order cone constraint. Thus, the lasso problem can be cast as
$$
 \displaystyle \min_{t,\beta^+,\beta^-,a_1,a_2,v}  \displaystyle \frac{t}{n} + \frac{\lambda}{n}\sum_{j=1}^p (\beta_j^++\beta_j^-) :
\begin{array}{l}
v = Y-X\beta^++X\beta^-,  t = -1 + 2a_1,  t =  1 + 2a_2\\
 (v,a_2,a_1)\in Q^{n+2}, \ t \geq 0, \beta^+ \in \RR^p_+, \ \beta^- \in \RR^p_+.\\
\end{array}
$$
Recall from the main text that the $\LASSO$ optimization problem can be cast similarly, but without  auxiliary variables $a_1, a_2$:
$$
 \displaystyle \min_{t,\beta^+,\beta^-,v}  \displaystyle \frac{t}{n^{1/2}} + \frac{\lambda}{n}\sum_{j=1}^p (\beta_j^++\beta_j^-) :
\begin{array}{rl}
& v = Y-X\beta^++X\beta^-\\
& (v,t)\in Q^{n+1},  \beta^+ \in \RR^p_+, \ \beta^- \in \RR^p_+.\\
\end{array}
$$
These conic formulations allow us to make several different computational methods directly applicable to compute these estimators.

{\bf First-order methods.}
Modern first-order methods focus on structured convex problems of the form:
 $$\min_w f\{A(w)+b\} + h(w)\ \ \mbox{or} \ \ \min_w h(w) \ : \ A(w) + b \in K,$$
 where $f$ is a smooth function and $h$ is a structured function that is possibly non-differentiable or having extended values. However it allows for an efficient proximal function to be solved, see `Templates for Convex Cone Problems with Applications to Sparse Signal Recovery' arXiv working paper 1009.2065 by Becker, Cand\`{e}s and Grant.
 By combining projections and subgradient information, these methods construct a sequence of iterates with strong theoretical guarantees. Recently these methods have been specialized for conic problems, which includes the lasso and the $\LASSO$ problems. %It is well known that several different formulations can % be made for the same optimization problem and the particular choice can impact the computational running times substantially. Below we focus on % % simple formulations for lasso and $\LASSO$.

Lasso is cast as $$\min_w f\{A(w)+b\} + h(w)$$ where $f(\cdot)=\|\cdot\|^2/n$, $h(\cdot)=(\lambda/n)\|\cdot\|_1$,  $A=X$, and $b=-Y$.
The projection required to be solved on every iteration for a given current point $\beta^k$ is
$$ \beta(\beta^{k}) = \arg \min_{\beta} 2\En\{x(y-x'\beta^{k})\}'\beta + \frac{1}{2}\mu\|\beta-\beta^k\|^2 + \frac{\lambda}{n}\|\beta\|_1,$$
where $\mu$ is a smoothing parameter. It follows that the minimization in $\beta$ above is separable and can be solved by soft-thresholding as
$$ \beta_j(\beta^{k}) = {\rm sign}\left[\beta_j^{k} + \frac{2\En\{x_{j}(y-x'\beta^{k})\}}{\mu} \right]\max\left[ \left|\beta_j^{k} + \frac{2\En\{x_{j}(y-x'\beta^{k})\}}{\mu}\right| - \frac{\lambda}{n\mu}, \ 0 \right].$$

For the $\LASSO$ the ``conic form" is  $$\min_w h(w) \ : \ A(w) + b \in K.$$
Letting $Q^{n+1}= \{ (z,t) \in \RR^n\times \RR \ : t\geq \|z\|\}$ and $h(w)=f(\beta,t)=t/n^{1/2}+(\lambda/n)\|\beta\|_1$ we have that
$$ \min_{\beta,t} \frac{t}{n^{1/2}} + \frac{\lambda}{n}\|\beta\|_1 : \ A(\beta,t) +b\in Q^{n+1} $$
where $b = (-Y',0)'$ and $A(\beta,t)\mapsto (\beta'X',t)'$.

In the associated dual problem, the dual variable $z\in\RR^n$ is constrained to be $\|z\|\leq 1/n^{1/2}$ (the corresponding dual variable associated with $t$ is set to $1/n^{1/2}$ to obtain a finite dual value). Thus we obtain
$$ \max_{\|z\|\leq 1/n^{1/2}} \inf_\beta \frac{\lambda}{n}\|\beta\|_1 + \frac{1}{2}\mu \|\beta-\beta^k\|^2-z'(Y-X\beta).$$
Given iterates $\beta^k, z^k$, as in the case of lasso, the minimization in $\beta$ is separable and can be solved by soft-thresholding as
$$ \beta_j(\beta^{k},z^k) = {\rm sign}\left\{\beta_j^k + (X'z^k/\mu)_j \right\}\max\left\{ \left|\beta_j^k + (X'z^k/\mu)_j\right| - \lambda/(n\mu), \ 0 \right\}.$$
The dual projection accounts for the constraint $\|z\|\leq 1/n^{1/2}$ and solves
$$z(\beta^k,z^k)=\arg \min_{\|z\|\leq 1/n^{1/2}} \frac{\theta_k}{2t_k}\|z-z_k\|^2 + (Y-X\beta^k)'z$$
which yields
$$ z(\beta^k,z^k) = \frac{z_k+(t_k/\theta_k)(Y-X\beta^k)}{\|z_k+(t_k/\theta_k)(Y-X\beta^k)\|}\min\left\{\frac{1}{n^{1/2}}, \ \|z_k+(t_k/\theta_k)(Y-X\beta^k)\|\right\}.$$

It is useful to note that, in the Tfocs package, the following command line computes the square-root lasso estimator: \\
\noindent {\it opts = tfocs\_SCD;}\\
\noindent {\it [ beta, out ] = tfocs\_SCD( prox\_l1(lambda/n), \{ X, -Y \}, proj\_l2(1/sqrt(n)), 1e-6 );}\\
where {n} denotes the sample size, {\rm lambda} the penalty choice, {\rm X} denote the {\rm n} by {\rm p} design matrix, and {\rm Y} a vector with $n$ observations of the response variable. The square-root lasso estimator is stored in the vector {\rm beta}.

{\bf Componentwise Search.} A common approach to solve unconstrained multivariate optimization problems  is to do componentwise
minimization, looping over components until convergence is achieved. This is particulary attractive in cases where the minimization over a single component can be done very efficiently.

Consider the following lasso optimization problem:
$$ \min_{\beta \in \RR^p } \En\{(y-x'\beta)^2\} + \frac{\lambda}{n}\sum_{j=1}^p\gamma_j|\beta_j|.$$
Under standard normalization assumptions we would have $\gamma_j=1$ and $\En(x_{j}^2) = 1$ ($j=1,\ldots,p$). The main
ingredient of the componentwise search for lasso is the rule that sets optimally the value of $\beta_j$ given fixed the values of the remaining variables:

For a current point $\beta$, let $\beta_{-j}=(\beta_1,\beta_2,\ldots,\beta_{j-1},0,\beta_{j+1},\ldots,\beta_p)'$: \\

If $2\En\{ x_{j}(y-x'\beta_{-j})\}  > \lambda\gamma_j/n $,  the optimal choice for $\beta_j$ is
$$ \beta_j = \left[- 2\En\{ x_{j}(y-x'\beta_{-j})\} + \lambda\gamma_j/n \right]/\En( x_{j}^2). $$

If $2\En\{ x_{j}(y-x'\beta_{-j})\}  < - \lambda\gamma_j/n $, the optimal choice for $\beta_j$ is
$$ \beta_j =  \left[2\En\{ x_{j}(y-x'\beta_{-j})\} - \lambda\gamma_j/n \right]/\En( x_{j}^2).$$

If $2|\En\{ x_{j}(y-x'\beta_{-j})\}|  \leq \lambda\gamma_j/n $, then $\beta_j = 0$. \\

This simple method is particularly attractive when the optimal solution is sparse which is typically the case of interest under choices of penalty levels that dominate the noise like $\lambda \geq cn\|S\|_\infty$.

Despite the additional square-root, which creates a non-separable criterion function, it turns out that the componentwise minimization for the $\LASSO$ also has a closed form solution.
Consider the following optimization problem:
$$ \min_{\beta \in \RR^p } \[\En\{(y-x'\beta)^2\}\]^{1/2} + \frac{\lambda}{n}\sum_{j=1}^p\gamma_j|\beta_j|.$$
As before, under standard normalization assumptions we would have $\gamma_j=1$ and $\En(x_{j}^2) = 1$ for $j=1,\ldots,p$.

The main
ingredient of the componentwise search for $\LASSO$ is the rule that sets optimally the value of $\beta_j$ given fixed the values of the remaining variables:

 If $\En\{ x_{j}(y-x'\beta_{-j})\}  > (\lambda/n) \gamma_j\{\widehat Q(\beta_{-j})\}^{1/2}$, set
{\small $$ \beta_j = - \frac{\En\{ x_{j}(y-x'\beta_{-j})]}{E_n(x_{j}^2)} + \frac{\lambda \gamma_j}{\En(x_{j}^2)}  \frac{\[ \widehat Q(\beta_{-j}) -  \{\En( x_{j}y-x_{j}x'\beta_{-j})\}^2 \{\En(x_{j}^2)\}^{-1} \]^{1/2}}{\[ n^2 - \{\lambda^2\gamma_j^2 / \En(x_{j}^2)\}  \]^{1/2}}.      $$}

 If $\En\{ x_{j}(y-x'\beta_{-j})\}  < - (\lambda/n) \gamma_j\{\widehat Q(\beta_{-j})\}^{1/2}$, set
{\small $$ \beta_j = - \frac{\En\{ x_{j}(y-x'\beta_{-j})\}}{\En(x_{j}^2)} -\frac{\lambda \gamma_j}{\En(x_{j}^2)}  \frac{\[ \widehat Q(\beta_{-j}) - \{\En( x_{j}y_i-x_{j}x'\beta_{-j})\}^2 \{\En(x_{j}^2)\}^{-1} \]^{1/2}}{[ n^2 - \{\lambda^2\gamma_j^2 / \En(x_{j}^2)\}  ]^{1/2}}.$$}

 If $|\En\{ x_{j}(y-x'\beta_{-j})\}|  \leq (\lambda/n) \gamma_j\{\widehat Q(\beta_{-j})\}^{1/2}$, set $\beta_j=0$.

\section{Additional Monte Carlo Results}\label{Sec:Empirical}

\subsection{Overview of Additional Monte Carlo Results}

In this section we provide more extensive Monte Carlo experiments to assess the finite sample performance of the proposed $\LASSO$ estimator. First we compare the performances of lasso and $\LASSO$ for different distributions of the noise and different designs. Second we compare $\LASSO$ with several feasible versions of lasso that estimate the unknown parameter $\sigma$.

\subsection{Detailed performance comparison of lasso and $\LASSO$}

Regarding the parameters for lasso and $\LASSO$, we set the penalty level according to the asymptotic options defined in the main text:
$$ \text{lasso penalty:} \ \sigma c \ 2n^{1/2}\Phi^{-1}(1-\alpha/2p) \ \ \ \LASSO \ \text{penalty:} \ c \ n^{1/2}\Phi^{-1}(1-\alpha/2p) $$
 respectively, with $1-\alpha=0.95$ and $c=1.1$.
As noted in the main text, experiments with the penalty levels according to the exact option led to similar behavior.

We use the linear regression model stated in the introduction of the main text as a data-generating process, with  either standard normal, $t(4)$, or asymmetric exponential errors:
$
\text{(a) } \ \epsilon_i \sim  N(0,1), \  \text{ (b) } \ \epsilon_i  \sim t(4)/2^{1/2}, \   \  \text{ or } \ \ \text{ (c) } \ \epsilon_i \sim \exp(1)-1 $
so that $\Ep(\epsilon^2_i) = 1$ in either case. We set the true parameter value as $
 \beta_0 = (1,1,1,1,1,0,\ldots,0)'$,  and we vary the parameter $\sigma$ between $0.25$ and $3$.  The number of regressors  is $p=500$,  the sample size is $n =100$, and we used $100$ simulations for each design. We generate  regressors as $x_i \sim N(0, \Sigma)$. We consider two design options for $\Sigma$:  Toeplitz correlation matrix $\Sigma_{jk}=(1/2)^{|j-k|}$ and equicorrelated correlation matrix $\Sigma_{jk}=(1/2)$.

\begin{figure}[!h]
\centering
\includegraphics[width=4in, height=2.5in]{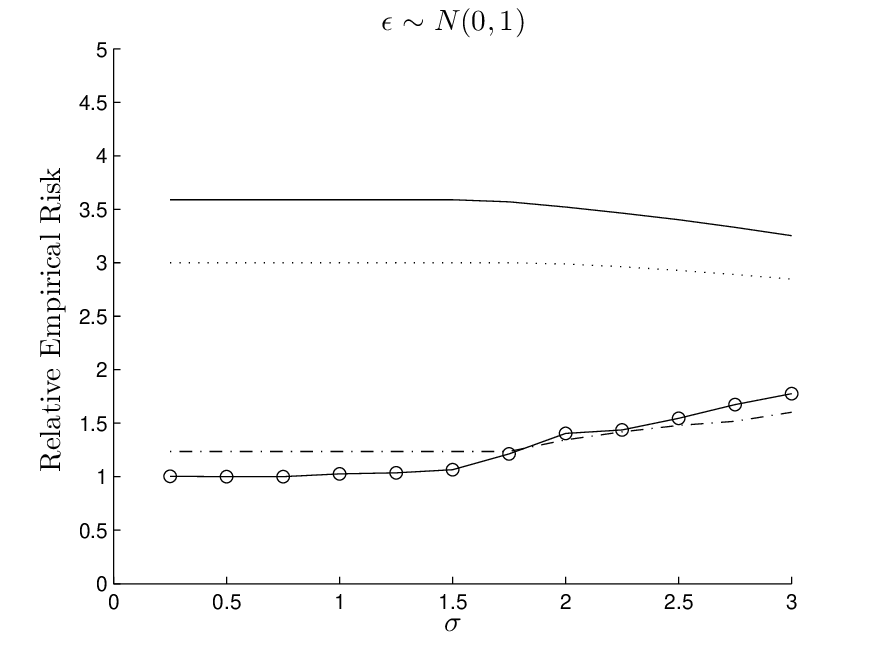}
\includegraphics[width=4in, height=2.5in]{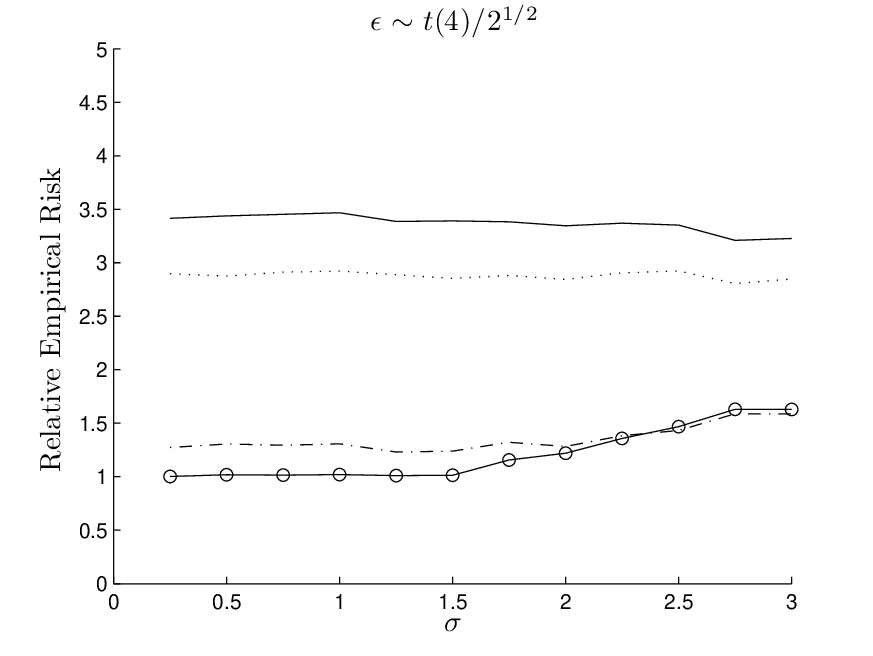}
\includegraphics[width=4in, height=2.5in]{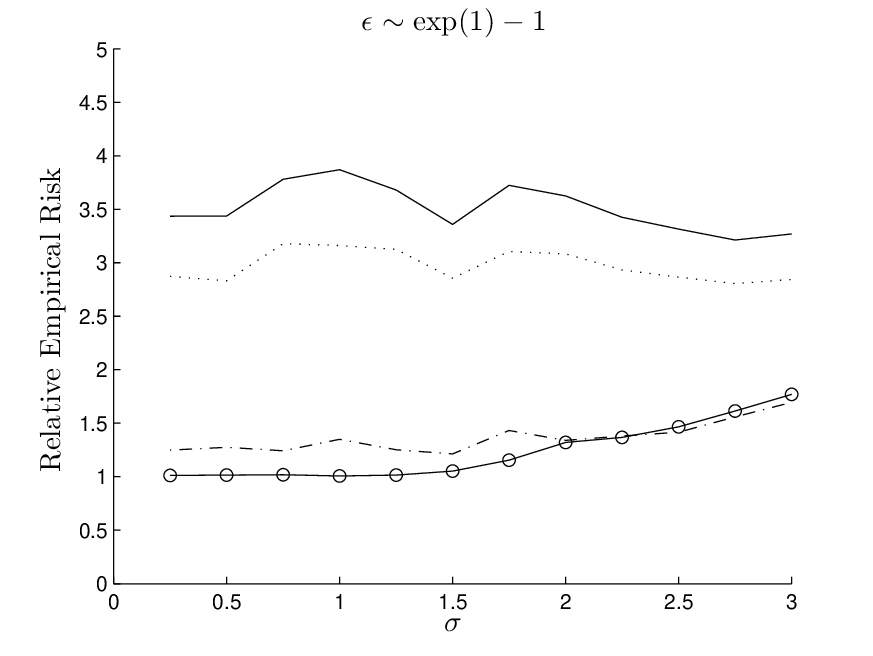}
\caption{Average relative empirical risk of infeasible lasso (dots), square-root lasso (solid), post infeasible lasso (dot-dash), and post square-root lasso (solid with circle),  with respect to the oracle estimator, that knows the true support, as a function of the standard deviation of the noise $\sigma$. In this experiment we used Toeplitz correlation matrix $\Sigma_{jk}=(1/2)^{|j-k|}$.}\label{Fig:MCcomp01}
\end{figure}

\begin{figure}[!h]
\centering
\includegraphics[width=3in, height=2in]{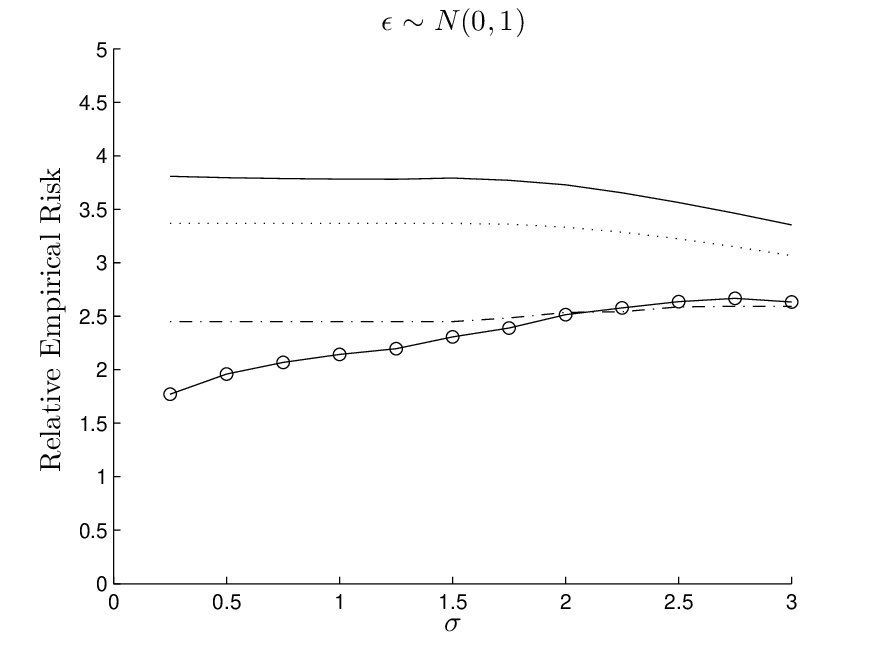}
\includegraphics[width=3in, height=2in]{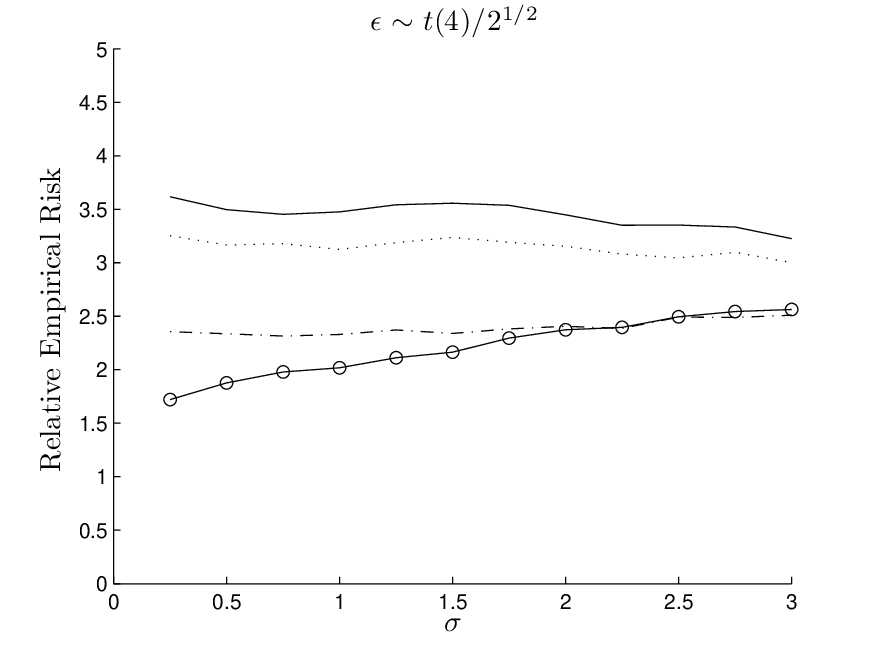}
\includegraphics[width=3in, height=2in]{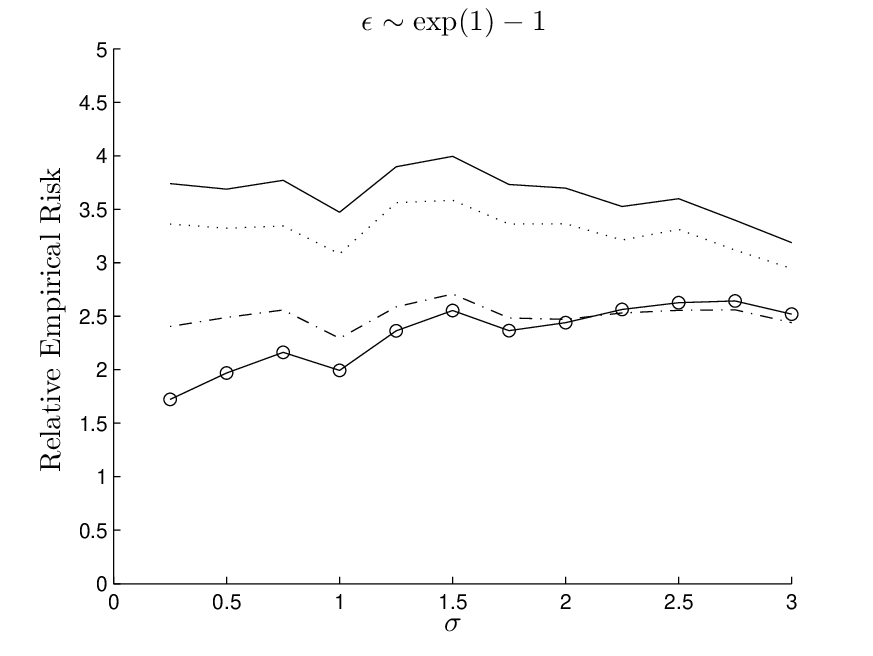}
\caption{Average relative empirical risk of infeasible lasso (dots), square-root lasso (solid), post infeasible lasso (dot-dash), and post square-root lasso (solid with circle),  with respect to the oracle estimator, that knows the true support, as a function of the standard deviation of the noise $\sigma$. In this experiment we used equicorrelated correlation matrix $\Sigma_{jk}=(1/2)$.}\label{Fig:MCcomp02}
\end{figure}

The results of computational experiments for designs a), b) and c) in Figures \ref{Fig:MCcomp01} and \ref{Fig:MCcomp01} illustrates the theoretical results indicated obtained in the paepr. That is, the performance of the non-Gaussian cases b) and c) is  very similar to the Gaussian case. Moreover, as expected, higher correlation between covariates translates into larger empirical risk.

The performance of $\LASSO$ and post $\LASSO$ are relatively close to the performance of lasso and post lasso that knows $\sigma$. These results are in close agreement with our theoretical results, which state that the  upper bounds on empirical
risk for $\LASSO$  asymptotically approach the analogous bounds for infeasible lasso.

\subsection{Comparison with feasible versions of lasso}

Next we focus on the Toeplitz design above to compare many traditional estimators related to lasso. More specifically we consider the following estimators: (1) oracle estimator, which is ols applied to the true minimal model (which is unknown outside the experiment),
(2) infeasible lasso with known $\sigma$ (which is unknown outside the experiment),
   (3)  post lasso, which applies ols to the model selected by infeasible lasso, (4)  $\LASSO$, (5) post $\LASSO$, which applies least squares
    to the model selected by $\LASSO$,
(6) 1-step feasible lasso, which is lasso with an estimate of $\sigma$ given by the conservative upper bound $\hat\sigma = [\En\{(y-\bar y)^2\}]^{1/2}$ where $\bar y = \En(y)$, (7) post 1-step lasso, which applies least squares to the model selected by 1-step lasso,
(8) 2-step lasso, which is lasso with an estimates of $\sigma$ given by the 1-step lasso estimator $\widetilde \beta$, namely $\hat \sigma = \{\widehat Q(\widetilde \beta)\}^{1/2}$, (9) post 2-step lasso, which applies least squares to the model selected by 2-step lasso,
(10) cv-lasso,  which is lasso with an estimate of $\lambda$ given by K-fold cross validation, (10) post cv-lasso, which applies OLS to the model selected by K-fold lasso, (11)  $\LASSO$ (1/2),    which uses the penalty of $\LASSO$ multiplied by $1/2$,
(12) post  $\LASSO$ (1/2), which applies least squares to the model selected by  $\LASSO$ (1/2).
We generate  regressors as $x_i \sim N(0, \Sigma)$ with the Toeplitz correlation matrix $\Sigma_{jk}=(1/2)^{|j-k|}$.% and with the equicorrelated design $\Sigma_{jk}=(1/2)$ for $j\neq k$, $1$ otherwise.

We focus our evaluation of the performance of an estimator $\widetilde \beta$ on the relative average empirical risk with respect to the oracle estimator $\beta^*$,  $\Ep(\|\widetilde \beta-\beta_0\|_{2,n})/\Ep(\| \beta^*-\beta_0\|_{2,n})$.

\begin{figure}[!h]
\centering Comparison of square-root lasso to cross-validation choice of $\lambda$ \\
\includegraphics[width=3in,height=2in]{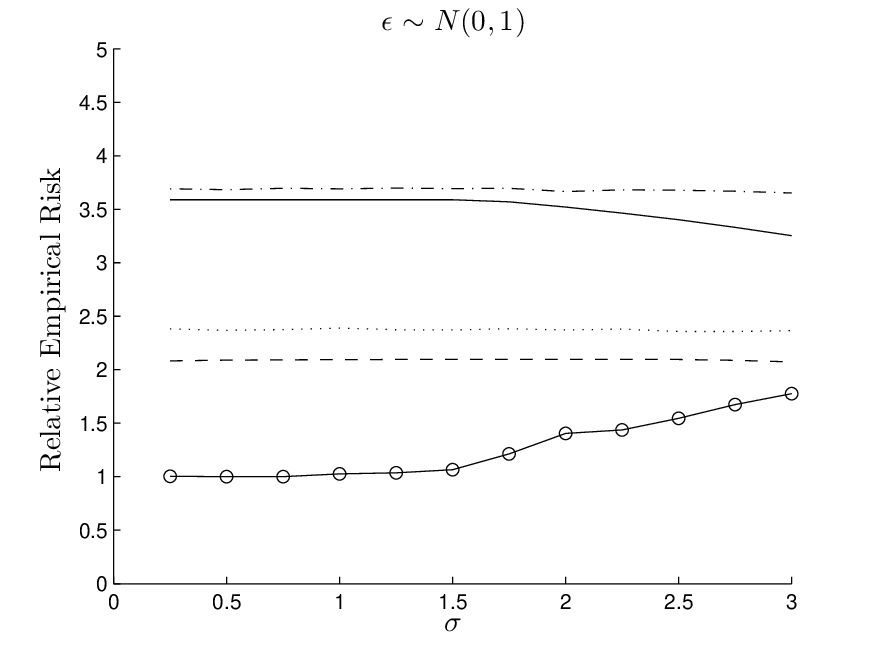}
\includegraphics[width=3in, height=2in]{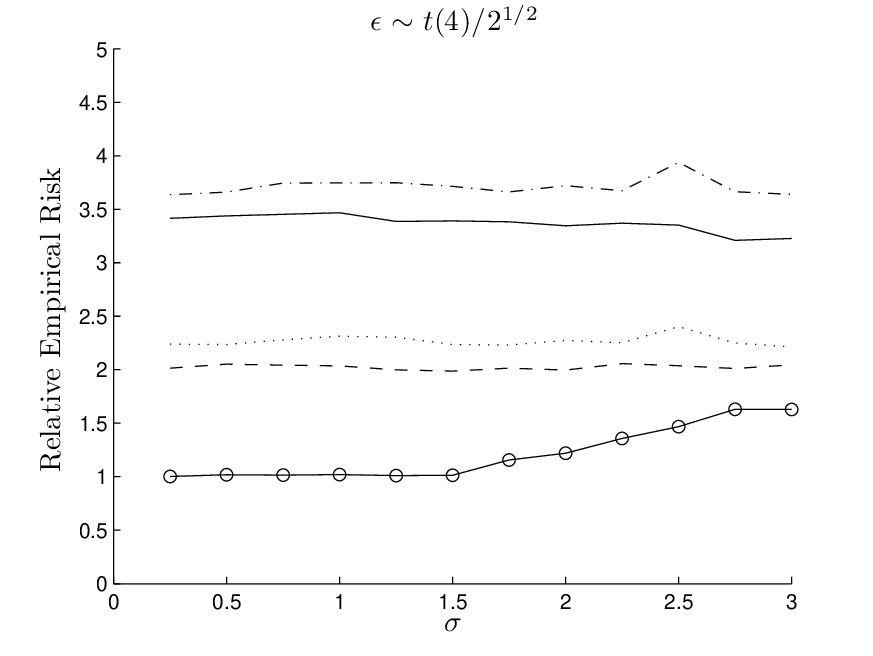}
\includegraphics[width=3in, height=2in]
{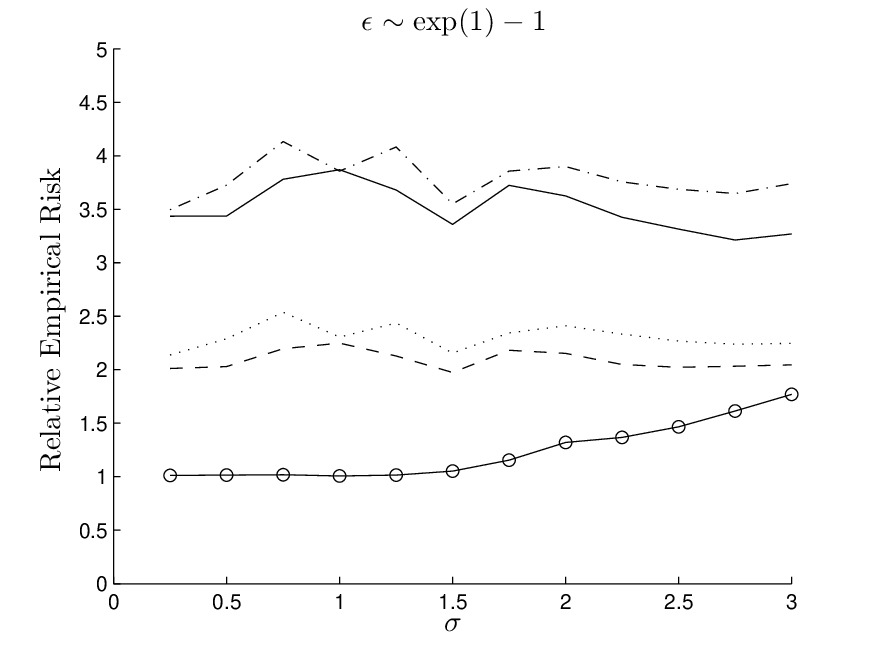}
\caption{Average relative empirical risk of square-root lasso (solid), post square-root lasso (solid with circle), cv-lasso (dots),  post cv-lasso (dot-dash), and square-root lasso (1/2) (dashes), with respect to the oracle estimator, that knows the true support, as a function of the standard deviation of the noise $\sigma$.
In this experiment we used Toeplitz correlation matrix $\Sigma_{jk}=(1/2)^{|j-k|}$.}\label{Fig:MCfirst01}
\end{figure}

 \begin{figure}[!h]
\centering  Comparison of square-root lasso to other feasible lasso methods \\
\includegraphics[width=3in,height=2in]{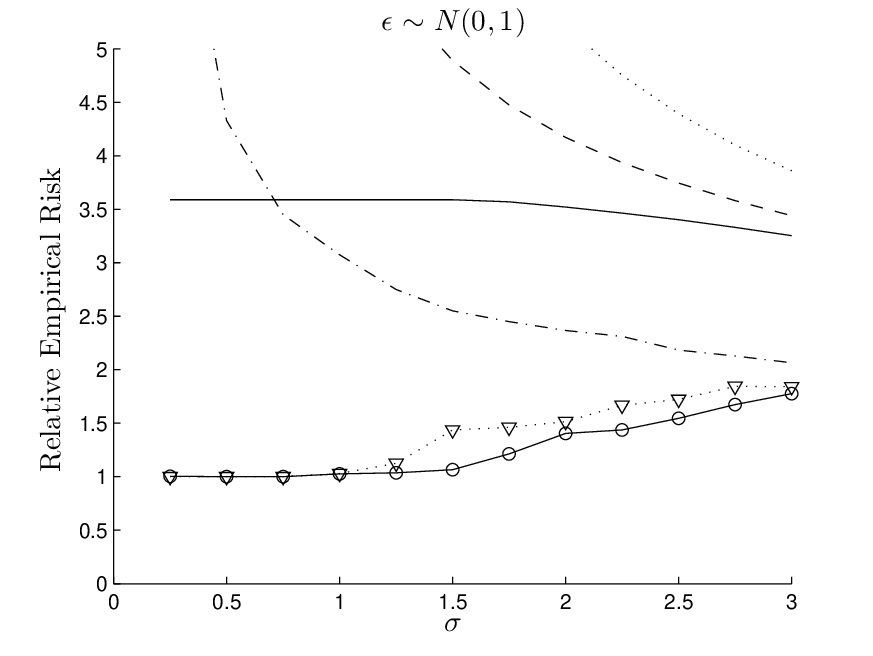}
\includegraphics[width=3in, height=2in]{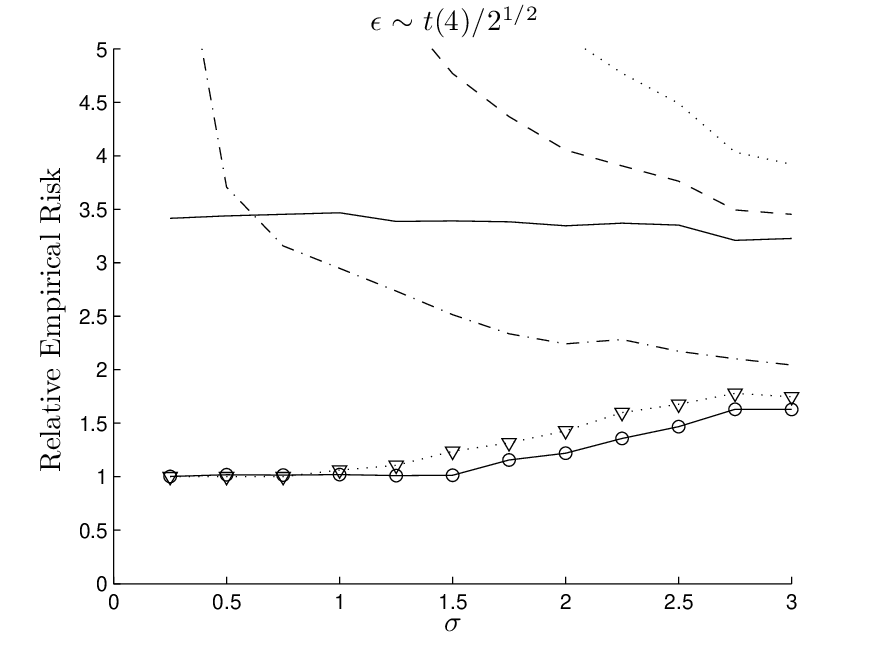}
\includegraphics[width=3in, height=2in]{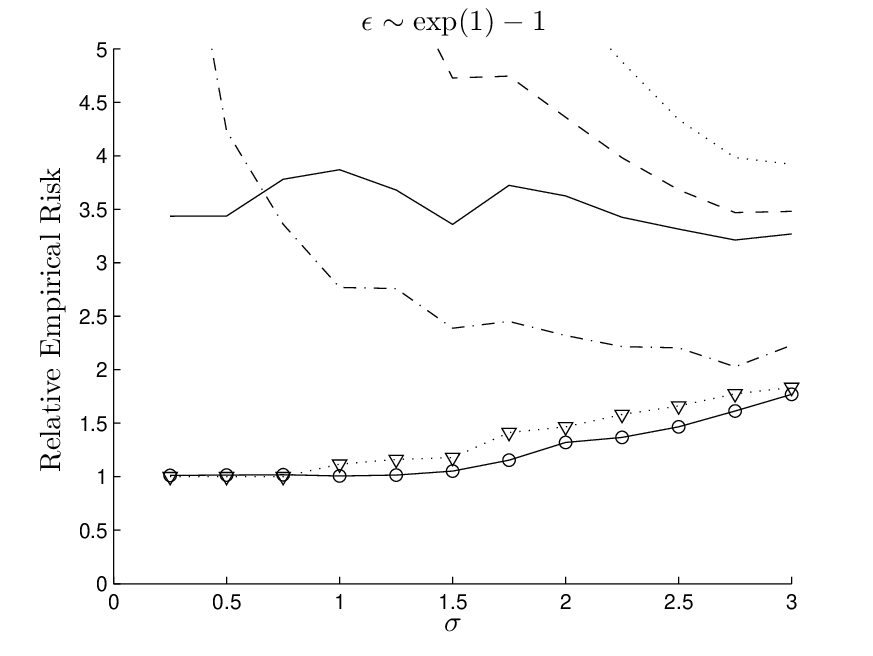}
\caption{%For different distributions of the noise, each graph shows the
Average relative empirical risk of square-root lasso (solid), post square-root lasso (solid with circle), 1-step lasso (dots),  post 1-step lasso (dot-dash), 2-step lasso (dashes), post 2-step lasso (dots with triangle), with respect to the oracle estimator, that knows the true support, as a function of the standard deviation of the noise $\sigma$. In this experiment we used Toeplitz correlation matrix $\Sigma_{jk}=(1/2)^{|j-k|}$.}\label{Fig:MCfirst02}
\end{figure}

%\begin{figure}[!h]
%\centering
%\includegraphics[width=6in, %height=2.5in]{SQRT_LASSO_SUPP_RelEmpRisk_normal.eps}
%\includegraphics[width=6in, %height=2.5in]{SQRT_LASSO_SUPP_RelEmpRisk_t4.eps}
%\includegraphics[width=6in, %height=2.5in]{SQRT_LASSO_SUPP_RelEmpRisk_exp.eps}
%\caption{The relative average empirical risk of the estimators as a %function of the standard deviation $\sigma$ of the noise. In this %experiment we used Toeplitz correlation matrix %$\Sigma_{jk}=(1/2)^{|j-k|}$.}\label{Fig:MCfirst01}
%\end{figure}

%\begin{figure}[!h]
%\centering
%\includegraphics[width=6in, height=3in]{SQRT_LASSO_SUPP_RelEmpRisk_normal.eps}
%\includegraphics[width=6in, height=3in]{SQRT_LASSO_SUPP_RelEmpRisk_t4.eps}
%\includegraphics[width=6in, height=3in]{SQRT_LASSO_SUPP_RelEmpRisk_exp.eps}
%\caption{The relative average empirical risk of the estimators as a function of the standard deviation $\sigma$ of the noise. In this experiment we used equicorrelated correlation matrix $\Sigma_{jk}=(1/2)$.}\label{Fig:MCfirst03}
%\end{figure}

%First we discuss the Toeplitz design.

We present the results comparing square-root lasso to lasso where the penalty parameter $\lambda$ is chosen based on $K$-fold cross-validation procedure. We report the experiments for designs a), b) and c) in Figure \ref{Fig:MCfirst01}. The first observation is that as indicated by theoretical results of the paper, the performance of the non-Gaussian cases b) and c) is  very similar to the Gaussian case so we focus on the later.

We observe that cv-lasso does improve upon square-root lasso (and infeasible lasso as well) with respect to empirical risk. The cross-validation procedure selects a smaller penalty level, which reduces the bias. However, cv-lasso is uniformly dominated by a square-root lasso method with penalty scaled by 1/2. Note the computational burden of cross-validation is substantial since one needs to solve several different lasso instances.
Importantly, cv-lasso does not perform well for purposes of model selection. This can be seen from the fact that post cv-lasso performs substantially worse than cv-lasso. Figure \ref{Fig:MCfirst01} also illustrates that square-root lasso performs substantially better than cv-lasso for purposes of models selection since post square-root lasso thoroughly dominates all other feasible methods considered.

Figure \ref{Fig:MCfirst02} compares other feasible lasso methods that are not as computational intense as cross-validation. The estimator with the best performance for all noise levels considered was the post $\LASSO$ reflecting the good model selection properties of the square-root lasso. The simple 1-step lasso with conservative estimate of $\sigma$ does very poorly.  The 2-step lasso does better, but it is still dominated by $\LASSO$. The post 1-step lasso and the post 2-step lasso are also  dominated by the post $\LASSO$ on all noise levels tested.

\section{Proofs of Additional Theoretical Results}

\begin{proof}[Proof of Lemma \ref{Lemma: non-normalMprimeRprime}] Part 1.  Let $t_n=\Phi^{-1}(1-\alpha/2p)$ and for some $w_n\to \infty$ slowly enough let $u_n = w_n \alpha^{-2/q}n^{-1/2} \log^{1/2} (n\vee p)  < 1/2$ for $n$ large enough.
Thus,
$$ \begin{array}{rl}
\Pr \( \Lambda > n^{1/2} t_n \mid X\) & \leq \Pr \[ \max_{1\leq j\leq p} \frac{n^{1/2}|\En(x_{j}\epsilon)|}{\{\En(x_{j}^2\epsilon^2)\}^{1/2}} > (1-u_n)t_n \mid X\] + \\
&  + \Pr\[\max_{1\leq j\leq p}\left\{\frac{\En(x_{j}^2\epsilon^2)}{\En(\epsilon^2)}\right\}^{1/2} > 1+u_n \mid X\], \\
\end{array}$$
 since $(1+u_n)(1-u_n)<1$. To bound the first term above, by Condition SN, we have that for $n$ large enough $t_n + 1\leq n^{1/6}/ [\ell_n\max_{1\leq j\leq p}\{\En(|x_{j}|^3)\Ep(|\epsilon_i|^3)\}^{1/3}]$ where $\ell_n\to \infty$ slowly enough. Thus, by the union bound and Lemma \ref{Lemma: MDSN}
{\small $$ \begin{array}{rl}
&\displaystyle \Pr \[ \max_{1\leq j\leq p} \frac{n^{1/2}|\En(x_{j}\epsilon)|}{\left\{\En(x_{j}^2\epsilon^2)\right\}^{1/2}} > (1-u_n)t_n \mid X\] \\ & \ \ \ \ \ \ \ \ \ \ \ \ \displaystyle  \leq p \max_{1\leq j\leq p}  \Pr \[ \frac{n^{1/2}|\En(x_{j}\epsilon)|}{\{\En(x_{j}^2\epsilon^2)\}^{1/2}} > (1-u_n)t_n \mid X\] \\
& \ \ \ \ \ \ \ \ \ \ \ \  \leq  2p\bar \Phi\{(1-u_n)t_n \} \left( 1 + \frac{A}{\ell_n^3}\right) \\
 &\ \ \ \ \ \ \ \ \ \ \ \  \leq \alpha \( 1 + \frac{1}{t_n^2}\)\frac{\exp(t_n^2u_n)}{1-u_n} \left( 1 + \frac{A}{\ell_n^3}\right) \\
\end{array}$$} where $t_n^2u_n=o(1)$ under condition SN, and the last inequality follows from standard bounds on  $\bar \Phi = 1 - \Phi$, and calculations similar to those in the proof of Lemma 1 of the main text. Moreover,
{\small $$ \begin{array}{rl}\displaystyle \Pr\[ \max_{1\leq j\leq p} \left\{\frac{|\En(x_{j}^2\epsilon^2)|}{\En(\epsilon^2)}\right\}^{1/2} > 1+u_n \mid X\]  &\displaystyle \leq \Pr\left\{ \max_{1\leq j\leq p} |\En(x_{j}^2\epsilon^2)|> 1+u_n  \mid X\right\} + \\
& + \Pr\left\{\En(\epsilon^2) < 1-(u_n/2)\right\} \\
\end{array}
$$}
since $1/(1+u_n) \leq 1-u_n+u_n^2 \leq 1-(u_n/2)$ since $u_n\leq 1/2$. It follows that
$$\begin{array}{rl}
 \Pr\{ \En(\epsilon^2) < 1-(u_n/2) \} & \leq \Pr\{ |\En(\epsilon^2)-1| > u_n/2 \} \\
 & \lesssim \alpha w_n^{-q/2}\log^{-q/4} (n\vee p) = o(\alpha)\end{array}$$
by the choice of $u_n$ and the application of Rosenthal's inequality.

Moreover, for $n$ sufficiently large, letting $\tau_1 = \tau_2 = \alpha / w_n^{1/2}$, we have
$$\begin{array}{rl}
u_n & = w_n\alpha^{-2/q}n^{-1/2}\log^{1/2}(p\vee n) \\
& \geq 4\left\{\frac{2\log(2p/\tau_1)}{n}\right\}^{1/2} \(\frac{\Ep(|\epsilon_i|^q)}{\tau_2}\)^{2/q} \max_{1\leq i\leq n}\|x_{i}\|_\infty^2\end{array}$$
by condition SN since we have $q>4$, $\max_{1\leq i\leq n}\|x_{i}\|_\infty$ is uniformly bounded above, $\log(2p/\tau_1)\lesssim \log(p\vee n)$,  and $w_n\to \infty$.
Thus, applying Lemma \ref{Lemma:ProcessSecond}, noting the relation above, we have
$$ \begin{array}{rl}
\Pr\(\max_{1\leq j\leq p}\En(x_{j}^2\epsilon^2) > 1+u_n \mid X\) & \displaystyle \leq \Pr\(\max_{1\leq j\leq p}|\En\{x_{j}^2(\epsilon^2-1)\}| > u_n \mid X\)\\
& \leq \tau_1 + \tau_2 = o(\alpha).\\
\end{array}$$

Part 2. Let $t_n=\Phi^{-1}(1-\alpha/2p)$ and for some $w_n\to \infty$ slowly enough let $u_n = w_n \alpha^{-2/q}n^{-1/2} \log^{1/2} (n\vee p)  < 1/2$ for $n$ large enough. Thus,
{\small $$ \begin{array}{rl}
\displaystyle \Pr \left\{ \Lambda > (1+u_n)(1+1/t_n)n^{1/2} t_n \mid X\right\} &\displaystyle  \leq \Pr \[
\max_{1\leq j\leq p} \frac{n^{1/2}|\En(x_{j}\epsilon)|}{\{\En(x_{j}^2\epsilon^2)\}^{1/2}} > (1+1/t_n)t_n \mid X\] + \\
&\displaystyle   + \Pr\[\max_{1\leq j\leq p}\left\{\frac{\En(x_{j}^2\epsilon^2)}{\En(\epsilon^2)}\right\}^{1/2} > 1+u_n \mid X\]. \\
\end{array}$$}
By the same argument as in part 1 we have
$$\Pr\[\max_{1\leq j\leq p}\left\{\frac{\En(x_{j}^2\epsilon^2)}{\En(\epsilon^2)}\right\}^{1/2} > 1+u_n \mid X\] = o(\alpha).$$

To bound the first term above, by Condition SN, we use that for $n$ large enough $(1+1/t_n)t_n + 1\leq n^{1/6}/ \ell_n\max_{1\leq j\leq p}\{\En(|x_{ij}|^3)\Ep(|\epsilon_i|^3)\}^{1/3}$ where $\ell_n\to \infty$ slowly enough. Thus, by the union bound and Lemma \ref{Lemma: MDSN}
\begin{align*}
\displaystyle &\Pr \[ \max_{1\leq j\leq p} \frac{n^{1/2}|\En(x_{j}\epsilon)|}{\left\{\En(x_{j}^2\epsilon^2)\right\}^{1/2}} > (1+1/t_n)t_n \mid X\] \displaystyle \\
& \leq p \max_{1\leq j\leq p}  \Pr \[ \frac{n^{1/2}|\En(x_{j}\epsilon)|}{\{\En(x_{j}^2\epsilon^2)\}^{1/2}} > t_n+1 \mid X\] \\
&\leq  2p\bar \Phi(t_n + 1) \left( 1 + \frac{A}{\ell_n^3}\right) \\
 & \leq \alpha \exp(-t_n-1/2)\frac{t_n}{t_n+1}\left\{ 1 + \frac{1}{(t_n+1)^2}\right\}\left( 1 + \frac{A}{\ell_n^3}\right)\\
 &  = \alpha \{ 1+o(1)\} \exp(-t_n-1/2), \\
\end{align*} where  the last inequality follows from standard bounds on $\bar \Phi$ and calculations  and calculations similar to those in the proof of Lemma 1 of the main text.

Therefore, for $n$ sufficiently large,
$$ \begin{array}{rl}
\Pr \{ \Lambda > (1+u_n)(1+1/t_n)n^{1/2} t_n \mid X\} & < \alpha \\
\end{array}$$ so that $\Lambda_{F_0}(1-\alpha\mid X) \leq  (1+u_n)(1+1/t_n)n^{1/2} t_n = \{1+o(1)\}n^{1/2} t_n$.
\end{proof}

\section{Omitted Proofs from the Main Text}

\subsection{Omitted Part of Proof of Lemma 1.}

Claim in the proof of Lemma 1: For independent random variables $\epsilon_i \sim N(0,1)$ $(i=1,\ldots,n)$ and any $0<r_n<1$, we have
$$ \Pr\{ \En(\epsilon^2) < (1-r_n)^2 \} \leq {\rm exp}(-nr_n^2/4).$$
It follows from
$$
\begin{array}{rl}
\Pr\{ \En(\epsilon^2) < (1-r_n)^2 \} & = \Pr\{ \En(\epsilon^2) < 1-2r_n+r_n^2 \} \\
& \leq \Pr\{ \En(\epsilon^2) < 1-r_n \} \\
& = \Pr\{ \En(\epsilon^2-1) < -r_n \} \\
& = \Pr\{ \sum_{i=1}^n(\epsilon^2_i-1) < -nr_n \} \\
& = \Pr\{ \sum_{i=1}^na_i(\epsilon^2_i-1) < -2|a|_2 \sqrt{n}r_n/2 \} \\
\end{array}
$$ where we have $a_i=1$ $(i=1,\ldots,n)$, so that $|a|_2 = \sqrt{n}$. Applying the second inequality of Lemma 1 of \cite{LaurentMassart2000} %(restated below for the reader's convenience)
for $\sqrt{x}=\sqrt{n}r_n/2$, we have
$$\Pr\{ \En(\epsilon^2) < (1-r_n)^2 \} \leq {\rm exp}(-nr_n^2/4).$$

%{\sc Lemma 1 of Laurent and Massart, 2000}. {\it
%Let $(Y_1,\ldots,Y_D)$ be i.i.d. Gaussian variables, with mean $0$ and variance $1$. Let $a_1,\ldots,a_D$ be nonnegative. We set
%$$ |a|_\infty = \sup_{i=1,\ldots,D}|a_i|, \ \ \ \ |a|_2^2 = \sum_{i=1}^D a_i^2.$$
%Let $$ Z = \sum_{i=1}^D a_i(Y_i^2-1).$$
%Then, the following inequalities hold for any positive $x$:
%$$ P( Z \geq 2|a|_2\sqrt{x}+2|a|_\infty x) \leq \exp(-x)$$
%$$ P( Z \leq -2|a|_2\sqrt{x}) \leq \exp(-x).$$}

\subsection{Omitted Part of Proof of Lemma 2.}

To show statement (iv) of Lemma 2, it suffices to show that for any $\nu'>1$, $\Pr ( c \Lambda > c \nu' n^{1/2} t_n \mid X  )  =  o(\alpha)$,
which follows analogously to the proof of statement (iii).

Indeed, for some constants $ 1 < \nu < \nu'$,
\begin{eqnarray*}
 && \Pr ( c \Lambda > c \nu' n^{1/2} t_n \mid X  ) \text {   } \nonumber \\
& & \leq_{(1)}   p \max_{1 \leq j \leq p} \Pr\{ | n^{1/2} \En (x_{j} \epsilon) | > t_n \nu \mid X \} + \Pr \[ \{\En(\epsilon^2)\}^{1/2} < \nu/\nu' \] \\
& & \leq   p \max_{1 \leq j \leq p} \Pr\{ | n^{1/2} \En (x_{j} \epsilon) | > t_n \nu \mid X \} + \Pr \{ \En(\epsilon^2) < (\nu/\nu')^2\} \\
& & \leq_{(2)}   p \max_{1 \leq j \leq p} \Pr\{ | n^{1/2} \En (x_{j} \epsilon) | > t_n \nu \mid  X \} + o(\alpha) \\
& &  =_{(3)}  2p \ \bar \Phi(t_n \nu  )\{1+ o(1)\} + o(\alpha)  =  o(\alpha)  \\
\end{eqnarray*}
where (1) holds by the union bound; (2) holds by the application
of the Rosenthal and Vonbahr-Esseen inequalities:
$$
\Pr [ \{\En(\epsilon^2)\}^{1/2} < \nu/\nu' ] \lesssim  n^{- \{\frac{q}{4}\wedge (\frac{q}{2}-1)\}} = o(\alpha)
$$
provided that
$$
\alpha^{-1} = o\{n^{ \frac{q}{4}\wedge (\frac{q}{2}-1)}\}.
$$
%(4)  by  $\phi(t)/t \sim \bar \Phi(t)$ as $t \to \infty$.
To verify relation (3), by Condition M and Slastnikov-Rubin-Sethuraman's
theorem on moderate deviations, we have that uniformly
in $ 0 \leq |t| \leq  k \log^{1/2} n$ for some $k^2 < q -2$, uniformly in $1 \leq j \leq p$ and for any $F =F_n \in \mathcal{F}$,
$\Pr\{ n^{1/2} |\En(x_{j} \epsilon)| > t\}/\{2\bar \Phi(t)\} \to 1$,
so the relation (3) holds to for $t = t_n \nu \leq  \{ 2 \log (2p/\alpha) \}^{1/2}  \nu \leq \nu\{\eta (q-2) \log n\}^{1/2}$ for
$\eta<1$ by assumption, provided $\nu$ is set sufficiently close to 1 so that $\nu^2\eta < 1$.

When $q> 4$, for large $n$ we can also bound $\Pr \{ |\En(\epsilon^2-1)|> (\nu/\nu')^2\}$ by $\Pr \{ |\En(\epsilon^2-1)|> r_n \}$
where $r_n = k \{\log n/n\}^{1/2}$, $k^2< q/2-2$,  and invoking the Slastnikov's theorem as previously, which gives
$
\Pr \{ |\En(\epsilon^2-1)|> (\nu/\nu')^2 \} \lesssim n^{-k^2} = o(\alpha) \text{ if } 1/\alpha = o(n^{k^2})= o(n^{\frac{q}{2}-2}).
$

Taking the best conditions on $1/\alpha$ gives the restriction:
\begin{equation}
\alpha^{-1} = o\left( \begin{array}{lll}  n^{q/2-1} & \text{ if } &  2 \leq q \leq 4 \\
 n^{q/4}  & \text{ if }  &  4 < q \leq 8 \\
 n^{q/2-2} & \text{ if } & 8 < q \\
\end{array} \right).
\end{equation}

\section{Tools Used}

\subsection {Rosenthal and Von Bahr-Esseen Inequalities}

\begin{lemma} Let $X_1,\ldots,X_n$ be independent zero-mean random variables, then for $r\geq 2$
$$
E\left(  \left\vert \sum_{i=1}^{n}X_{i}\right\vert ^{r}\right) \leq C(r)  \max\left[ \sum_{t=1}^n \Ep(|X_i|^r),   \left\{\sum_{t=1}^n \Ep(X_i^2)\right\}^{r/2} \right].
$$
\end{lemma}

This is due to \cite{Rosenthal1970}.

\begin{corollary} Let $r \geq 2$, and consider the case of identically distributed zero-mean variables $X_i$ with $\Ep(X_i^2)=1$ and $\Ep(|X_i|^r)$ bounded by $C$.
Then for any $\ell_n \to \infty$
$$
Pr \left(\frac{ |\sum_{i=1}^n X_i|}{n} >  \ell_n n^{-1/2} \right) \leq  \frac{2C(r)C}{\ell_n^{r}} \to 0.
$$
\end{corollary}

To verify the corollary, we use Rosenthal's inequality $E\left(  \left\vert \sum_{i=1}^{n}X_{i}\right\vert ^{r}\right) \leq C n^{r/2}$,
and the result follows by Markov inequality, $$ P\left(\frac{|\sum_{i=1}^n X_i|}{n} >c\right) \leq  \frac{ C(r) C n^{r/2} }{c^rn^r} \leq  \frac{C(r) C}{ c^r n^{r/2} }.$$

%Reference: H.P. Rosenthal, 1970. On the subspaces of $L^p$ $(p>2)$ spanned by sequences of independent random variables, Israel J. Math. 8 (1970), pp. 273–-303.

\begin{lemma} Let $X_1,\ldots,X_n$ be independent zero-mean random variables. Then for $1 \leq r \leq 2$
$$
\Ep \left( \left|  \sum_{i=1}^n X_i\right|^r\right) \leq  (2- n^{-1}) \cdot \sum_{k=1}^n \Ep(|X_k|^r).
$$
\end{lemma}

This result is due to  \cite{vonBahrEsseen1965}.

%Reference: Theorem 3 in VE in von Bahr, Bengt; Esseen, Carl-Gustav Inequalities for the rth absolute moment of a sum of random variables, $1 \leq r \leq 2$. Ann. Math. Statist 36, 1965, pp. 299--303.

%Remark. Vonbahr-Esseen actually applies to martingale differences.

\begin{corollary}  Let $r \in [1,2]$, and consider the case of identically distributed zero-mean variables $X_i$ with $\Ep(|X_i|^r)$ bounded by $C$.
Then for any $\ell_n \to \infty$
$$
\Pr \left\{\frac{ \left|\sum_{i=1}^n X_i\right|}{n} >  \ell_n n^{-(1-1/r)} \right\}  \leq \frac{2C}{\ell_n^r} \to 0.
$$
\end{corollary}

The corollary follow by Markov and Vonbahr-Esseen's inequalities,
$$
\Pr\left(\frac{|\sum_{i=1}^n X_i|}{n} >c\right) \leq  \frac{ C \Ep \left(|  \sum_{i=1}^n X_i|^r\right) }{c^rn^r} \leq  \frac{ n \Ep(|X_i|^r)}{ c^r n^r} \leq  C\frac{\Ep(|X_i|^r)}{ c^r n^{r-1}}.
$$

\subsection{Moderate Deviations for Sums}
Let $X_{ni}, i=1,\ldots,n; n \geq 1$ be a double sequence of row-wise independent random variables with $\Ep( X_{ni}) =0$, $\Ep( X^2_{ni})< \infty$, $i=1,\ldots,k_n$; $n\geq 1$, and $B_n^2 = \sum_{i=1}^{k_n} \Ep(X_{ni}^2) \to \infty$ as $n \to \infty$. Let
$$
F_n(x) = \Pr \left( \sum_{i=1}^{k_n} X_{in}< x B_n \right).
$$

The following result is due to \cite{Slastnikov1979}.
\begin{lemma} \textit{ If for sufficiently large $n$ and some positive constant $c$,
$$
\sum_{i=1}^{k_n} \Ep\{|X_{ni}|^{2+c^2} \rho(|X_{ni}|) \log^{ -(1+c^2)/2}(3+|X_{ni}|)\} \leq g(B_n) B_n^2,
$$
where $\rho(t)$ is slowly varying function monotonically growing to infinity and $g(t) = o\{\rho(t)\}$ as $t \to \infty$, then
$$
1 - F_n(x) \sim 1- \Phi(x), \ \ \ \ F_n(-x) \sim \Phi(-x), \ \ \   n \to \infty,
$$
uniformly in  the region $0 \leq x \leq c \{\log B_n^2\}^{1/2}$.}
\end{lemma}

The following result is due to \cite{Slastnikov1979} and \cite{rubin:sethuraman}.

\begin{corollary}\textit{ If $q > c^2 + 2$ and
$
\sum_{i=1}^{k_n} \Ep(|X_{ni}|^{q}) \leq K B_n^2,
$
then
$$
1-F_n(x) \sim 1- \Phi(x), \ \ \  F_n(-x) \sim \Phi(-x), \ \ \   n \to \infty,
$$
uniformly in  the region $0 \leq x \leq c \{\log B_n^2\}^{1/2}$.}
\end{corollary}

Remark. Rubin-Sethuraman derived the corollary for $x= t\{\log B_n^2\}^{1/2}$ for fixed $t$. Slastnikov's result adds uniformity
and relaxes the moment assumption.

%Reference:    A. D. Slastnikov  Theory of Probability and its Applications, Volume 23, 1979. Limit theorems for moderate deviation  probabilities, pp. 322--340

%\textbf{Marcinkiewicz-Zygmund Inequality.}

%This is the inequality that is a main input in the proof of Marcinkiewicz-Zygmund's LLN.

%Suppose $ x_{i}$ $i=1,\ldots,n$ are independent random variables such that $ E\left( x_{i}\right)  =0$ and $ E\left(  \left\vert x_{i}\right\vert ^{p}\right) <+\infty$, $ 1\leq p<+\infty$,
%
%$ A_{p} \leq E\left(  \left(  \sum_{i=1}^{n}\left\vert x_{i}\right\vert ^{2}\right) _{{}}^{p/2}\right)  \leq E\left(  \left\vert \sum_{i=1}^{n}x_{i}\right\vert ^{p}\right)  \leq B_{p}E\left(  \left(  \sum_{i=1}^{n}\left\vert x_{i}\right\vert ^{2}\right)  _{{}}^{p/2}\right) $

%where $ A_{p}$ and $ B_{p}$ are positive constants, which depend only on $ p$.

%J. Marcinkiewicz and A. Zygmund. Sur les foncions independantes. Fund. Math., 28:60–90, 1937. Reprinted in Józef Marcinkiewicz, Collected papers, edited by Antoni Zygmund, Panstwowe Wydawnictwo Naukowe, Warsaw, 1964, pp. 233–259.

%Yuan Shih Chow and Henry Teicher. Probability theory. Independence, interchangeability, martingales. Springer-Verlag, New York, second edition, 1988.

\subsection{Moderate Deviations for Self-Normalizing Sums}

We shall be using the following two technical results. The first follows from Theorem 7.4 in \cite{PLS2009} which is based on \cite{JingShaoWang2003}.

\begin{lemma} \label{Lemma: MDSN} Let $X_{1,n},\ldots,X_{n,n}$ be the triangular array of independent non-identically distributed zero-mean random variables. Suppose $$M_n = \frac{ \{\frac{1}{n}\sum_{i=1}^n\Ep(X_{i,n}^2)\}^{1/2}}{ \{\frac{1}{n}\sum_{i=1}^n\Ep(|X_{i,n}|^3)\}^{1/3} }>0 \ \ \mbox{
and for some $\ell_n \to \infty$ we have } \ \
 n^{1/6} M_n/\ell_n \geq 1.
$$
Then there is a universal constant $A$ such that uniformly on
$
0 \leq  x \leq  n^{1/6} M_n/\ell_n-1$, the
quantities
$$
S_{n,n} = \sum_{i=1}^n X_{i,n},  \ \ V^2_{n,n} = \sum_{i=1}^n X^2_{i,n}$$
obey
$$
\left |\frac{\Pr(|S_{n,n}/V_{n,n}|  \geq x) }{ 2 \bar \Phi(x)} - 1 \right |  \leq
\frac{A}{ \ell_n^3} \to 0.
$$
\end{lemma}

The second follows from the proof of Lemma 10 given in the working paper ``Pivotal Estimation of Nonparametric Functions via Square-root Lasso", arXiv 1105.1475v2, by the authors, which is based on symmetrization arguments.

\begin{lemma}\label{Lemma:ProcessSecond}
Let $\epsilon_i$ ($i=1,\ldots,n$) be independent identically distributed random variables such that $\Ep(\epsilon_i^2)=1$ and $\sup_{n\geq 1}\Ep(|\epsilon_i|^q)<\infty$ for $q\geq 4$. Conditional on $x_1,\ldots,x_n \in \RR^p$, with probability $1-4\tau_1-4\tau_2$
$$\max_{1\leq j\leq p}|\En\{x_{j}^2(\epsilon^2-1)\}| \leq 4\left\{\frac{2\log(2p/\tau_1)}{n}\right\}^{1/2} \left\{\frac{\Ep(|\epsilon_i|^q)}{\tau_2}\right\}^{2/q} \max_{1\leq i\leq n}\|x_{i}\|_\infty^2.$$
\end{lemma}

\bibliographystyle{plain}
\bibliography{biblioSqLASSO}

\end{document}